\documentclass{article}
\pdfoutput=1

\usepackage{microtype}
\usepackage{booktabs} % for professional tables
\usepackage{epsf}
\usepackage{epsfig}
\usepackage{subfigure}
\usepackage{amsmath}
\usepackage{amssymb}
\usepackage{amsfonts}
\usepackage{amsthm}
\usepackage{mathtools}
\usepackage{multirow}
\usepackage{multicol}
\usepackage{graphicx} % more modern
\usepackage{xspace}
\usepackage{listings}
\usepackage{xcolor}
\usepackage{enumitem}

\usepackage{hyperref}

\usepackage[accepted]{icml2018}

\begin{document}

\def\Blue{\color{blue}}
\def\Purple{\color{purple}}

\def\A{{\bf A}}
\def\a{{\bf a}}
\def\B{{\bf B}}
\def\C{{\bf C}}
\def\c{{\bf c}}
\def\D{{\bf D}}
\def\d{{\bf d}}
\def\F{{\bf F}}
\def\e{{\bf e}}
\def\f{{\bf f}}
\def\G{{\bf G}}
\def\H{{\bf H}}
\def\I{{\bf I}}
\def\K{{\bf K}}
\def\L{{\bf L}}
\def\M{{\bf M}}
\def\m{{\bf m}}
\def\N{{\bf N}}
\def\n{{\bf n}}
\def\Q{{\bf Q}}
\def\q{{\bf q}}
\def\S{{\bf S}}
\def\s{{\bf s}}
\def\T{{\bf T}}
\def\U{{\bf U}}
\def\u{{\bf u}}
\def\V{{\bf V}}
\def\v{{\bf v}}
\def\W{{\bf W}}
\def\w{{\bf w}}
\def\X{{\bf X}}
\def\x{{\bf x}}
\def\Y{{\bf Y}}
\def\y{{\bf y}}
\def\Z{{\bf Z}}
\def\z{{\bf z}}
\def\0{{\bf 0}}
\def\1{{\bf 1}}

\def\AM{{\mathcal A}}
\def\CM{{\mathcal C}}
\def\DM{{\mathcal D}}
\def\GM{{\mathcal G}}
\def\FM{{\mathcal F}}
\def\IM{{\mathcal I}}
\def\NM{{\mathcal N}}
\def\OM{{\mathcal O}}
\def\SM{{\mathcal S}}
\def\TM{{\mathcal T}}
\def\UM{{\mathcal U}}
\def\XM{{\mathcal X}}
\def\YM{{\mathcal Y}}
\def\RB{{\mathbb R}}

\def\TX{\tilde{\bf X}}
\def\tx{\tilde{\bf x}}
\def\ty{\tilde{\bf y}}
\def\TZ{\tilde{\bf Z}}
\def\tz{\tilde{\bf z}}
\def\hd{\hat{d}}
\def\HD{\hat{\bf D}}
\def\hx{\hat{\bf x}}
\def\TD{\tilde{\Delta}}

\def\alp{\mbox{\boldmath$\alpha$\unboldmath}}
\def\bet{\mbox{\boldmath$\beta$\unboldmath}}
\def\epsi{\mbox{\boldmath$\epsilon$\unboldmath}}
\def\etab{\mbox{\boldmath$\eta$\unboldmath}}
\def\ph{\mbox{\boldmath$\phi$\unboldmath}}
\def\pii{\mbox{\boldmath$\pi$\unboldmath}}
\def\Ph{\mbox{\boldmath$\Phi$\unboldmath}}
\def\Ps{\mbox{\boldmath$\Psi$\unboldmath}}
\def\tha{\mbox{\boldmath$\theta$\unboldmath}}
\def\Tha{\mbox{\boldmath$\Theta$\unboldmath}}
\def\muu{\mbox{\boldmath$\mu$\unboldmath}}
\def\Si{\mbox{\boldmath$\Sigma$\unboldmath}}
\def\si{\mbox{\boldmath$\sigma$\unboldmath}}
\def\Gam{\mbox{\boldmath$\Gamma$\unboldmath}}
\def\Lam{\mbox{\boldmath$\Lambda$\unboldmath}}
\def\De{\mbox{\boldmath$\Delta$\unboldmath}}
\def\Ome{\mbox{\boldmath$\Omega$\unboldmath}}
\def\TOme{\mbox{\boldmath$\hat{\Omega}$\unboldmath}}
\def\vps{\mbox{\boldmath$\varepsilon$\unboldmath}}
\newcommand{\ti}[1]{\tilde{#1}}
\def\Ncal{\mathcal{N}}
\def\argmax{\mathop{\rm argmax}}
\def\argmin{\mathop{\rm argmin}}
\providecommand{\abs}[1]{\lvert#1\rvert}
\providecommand{\norm}[2]{\lVert#1\rVert_{#2}}

\def\Zs{{\Z_{\mathrm{S}}}}
\def\Zl{{\Z_{\mathrm{L}}}}
\def\Yr{{\Y_{\mathrm{R}}}}
\def\Yg{{\Y_{\mathrm{G}}}}
\def\Yb{{\Y_{\mathrm{B}}}}
\def\Ar{{\A_{\mathrm{R}}}}
\def\Ag{{\A_{\mathrm{G}}}}
\def\Ab{{\A_{\mathrm{B}}}}
\def\As{{\A_{\mathrm{S}}}}
\def\Asr{{\A_{\mathrm{S}_{\mathrm{R}}}}}
\def\Asg{{\A_{\mathrm{S}_{\mathrm{G}}}}}
\def\Asb{{\A_{\mathrm{S}_{\mathrm{B}}}}}
\def\Or{{\Ome_{\mathrm{R}}}}
\def\Og{{\Ome_{\mathrm{G}}}}
\def\Ob{{\Ome_{\mathrm{B}}}}

\def\vec{\mathrm{vec}}
\def\fold{\mathrm{fold}}
\def\index{\mathrm{index}}
\def\sgn{\mathrm{sgn}}
\def\tr{\mathrm{tr}}
\def\rk{\mathrm{rank}}
\def\diag{\mathsf{diag}}
\def\const{\mathrm{Const}}
\def\dg{\mathsf{dg}}
\def\st{\mathsf{s.t.}}
\def\vect{\mathsf{vec}}
\def\MCAR{\mathrm{MCAR}}
\def\MSAR{\mathrm{MSAR}}
\def\etal{{\em et al.\/}\,}
\newcommand{\indep}{{\;\bot\!\!\!\!\!\!\bot\;}}

\def\Lsize{\hbox{\space \raise-2mm\hbox{$\textstyle \L \atop \scriptstyle {m\times 3n}$} \space}}
\def\Ssize{\hbox{\space \raise-2mm\hbox{$\textstyle \S \atop \scriptstyle {m\times 3n}$} \space}}
\def\Osize{\hbox{\space \raise-2mm\hbox{$\textstyle \Ome \atop \scriptstyle {m\times 3n}$} \space}}
\def\Tsize{\hbox{\space \raise-2mm\hbox{$\textstyle \T \atop \scriptstyle {3n\times n}$} \space}}
\def\Bsize{\hbox{\space \raise-2mm\hbox{$\textstyle \B \atop \scriptstyle {m\times n}$} \space}}

\newcommand{\twopartdef}[4]
{
	\left\{
		\begin{array}{ll}
			#1 & \mbox{if } #2 \\
			#3 & \mbox{if } #4
		\end{array}
	\right.
}

\newcommand{\tabincell}[2]{\begin{tabular}{@{}#1@{}}#2\end{tabular}}

\newtheorem{theorem}{Theorem}
\newtheorem{lemma}{Lemma}
\newtheorem{proposition}{Proposition}
\newtheorem{corollary}{Corollary}
\newtheorem{definition}{Definition}
\newtheorem{remark}{Remark}

\def\E{{\mathbb E}}
\def\R{{\mathbb R}}

\DeclarePairedDelimiter\ceil{\lceil}{\rceil}
\DeclarePairedDelimiter\floor{\lfloor}{\rfloor}

\newcommand{\ip}[2]{\left\langle #1, #2 \right \rangle}

\twocolumn[
\icmltitle{Generalized Byzantine-tolerant SGD}



\icmlsetsymbol{equal}{*}

\begin{icmlauthorlist}
\icmlauthor{Cong Xie}{UIUC}
\icmlauthor{Oluwasanmi Koyejo}{UIUC}
\icmlauthor{Indranil Gupta}{UIUC}
\end{icmlauthorlist}

\icmlaffiliation{UIUC}{Department of Computer Science, University of Illinois at Urbana-Champaign, Illinois, US}
%\icmlaffiliation{goo}{Googol ShallowMind, New London, Michigan, USA}
%\icmlaffiliation{ed}{School of Computation, University of Edenborrow, Edenborrow, United Kingdom}
%
\icmlcorrespondingauthor{Cong Xie}{cx2@illinois.edu}
%\icmlcorrespondingauthor{Eee Pppp}{ep@eden.co.uk}

\icmlkeywords{Byzantine failure, marginal median}

\vskip 0.3in
]




\begin{abstract}
%We study the high-dimensional median as robust aggregation for distributed synchronous Stochastic Gradient Descent~(SGD). To be more specific, we study parameter server architecture, in which the workers have  arbitrary~(Byzantine) failures. We propose two Byzantine-resilient gradient aggregation rules with provable convergence. Furthermore, using marginal median, we propose the first dimensionally Byzantine-resilient algorithm for distributed synchronous SGD. Empirical results also show good performance. 

We propose three new robust aggregation rules for distributed synchronous Stochastic Gradient Descent~(SGD) under a general Byzantine failure model. The attackers can arbitrarily manipulate the data transferred between the servers and the workers in the parameter server~(PS) architecture. We prove the Byzantine resilience properties of these aggregation rules. Empirical analysis shows that the proposed techniques outperform current approaches for realistic use cases and Byzantine attack scenarios.
\end{abstract}

\section{Introduction}

The failure resilience of distributed machine-learning systems has attracted increasing attention~\citep{blanchard2017machine,chen2017distributed} in the community. 
Larger clusters can accelerate training. However, this makes the distributed system more vulnerable to different kinds of failures or even attacks, including crashes and computation errors, stalled processes, or compromised sub-systems~\citep{harinath2017review}. 
Thus, failure/attack resilience is becoming more and more important for distributed machine-learning systems, especially for large-scale deep learning~\citep{Dean2012LargeSD,McMahan2017CommunicationEfficientLO}.  

In this paper, we consider the most general failure model, Byzantine failures~\citep{Lamport1982TheBG}, where the attackers can know any information of the other processes, and attack any value in transmission. To be more specific,  the data transmission between the machines can be replaced by arbitrary values. Under such model, there are no constraints on the failures or attackers. 
%A variable is Byzantine if it is replaced by arbitrary value. 

%in which a subset of processes can generate completely arbitrary values.

The distributed training framework studied in this paper is the \underline{P}arameter \underline{S}erver~(PS). The PS architecture is composed of the server nodes and the worker nodes. The server nodes maintain a global copy of the model, aggregate the gradients from the workers, apply the gradients to the model, and broadcast the latest model to the workers. The worker nodes pull the latest model from the server nodes, compute the gradients according to the local portion of the training data, and send the gradients to the server nodes. The entire dataset and the corresponding workload is distributed to multiple worker nodes, thus parallelizing the computation via partitioning the dataset. 
There exist several distributed machine learning systems using the PS architecture. For instance, Tensorflow~\cite{Abadi2016TensorFlowAS}, CNTK~\cite{Seide2016CNTKMO}, and MXNet~\cite{Chen2015MXNetAF} implement internal PS's.

In this paper, we study the Byzantine resilience of synchronous Stochastic Gradient Descent~(SGD), which is a popular class of learning algorithms using PS architecture. Its variants are widely used in training deep neural networks~\cite{Kingma2014AdamAM,Mukkamala2017VariantsOR}. Such algorithms always wait to collect gradients from all the worker nodes before moving on to the next iteration. 

\begin{figure}[htb!]
\centering
\subfigure[Classic Byzantine]{\includegraphics[width=0.238\textwidth]{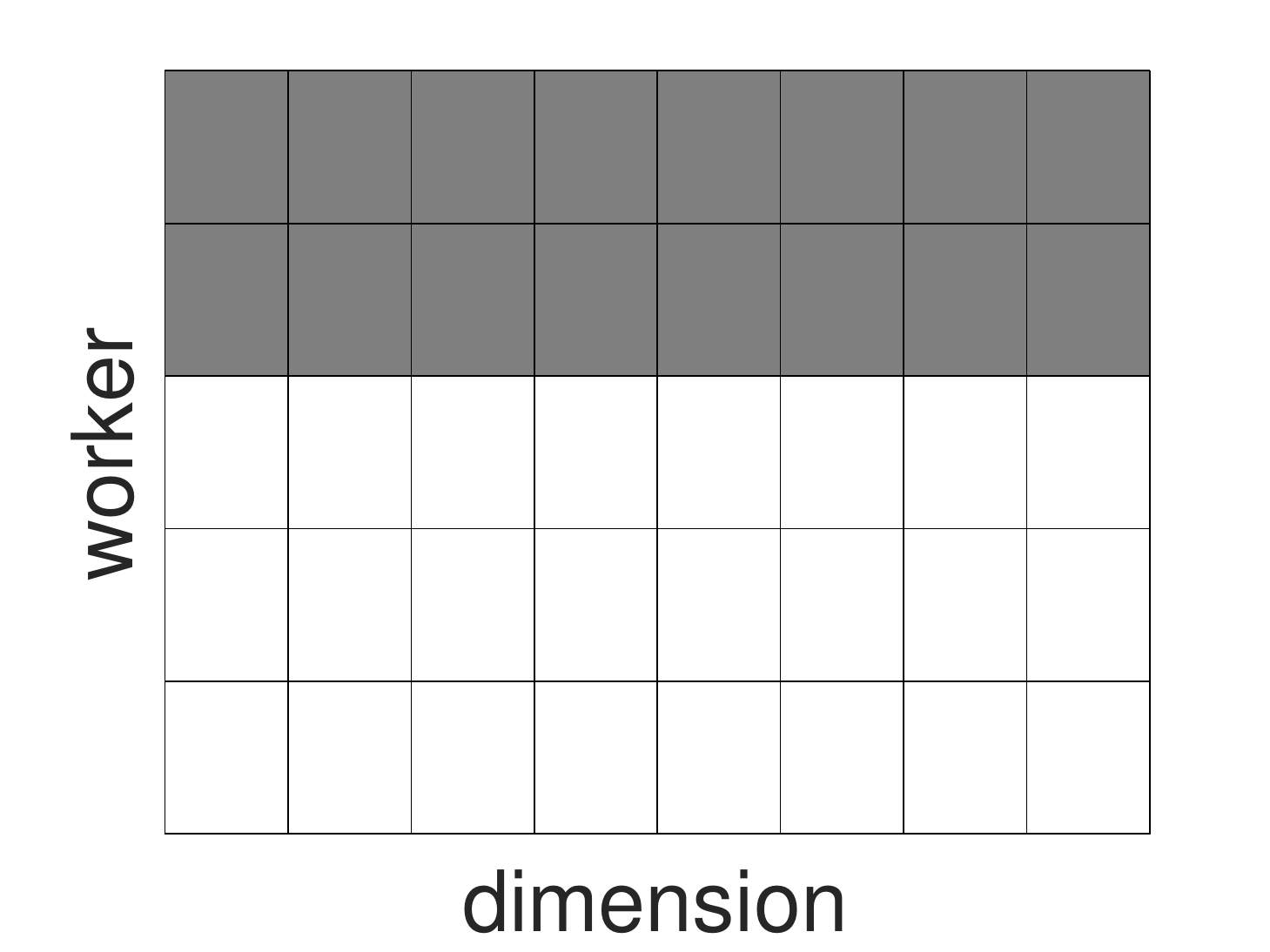}}
\subfigure[Generalized Byzantine]{\includegraphics[width=0.238\textwidth]{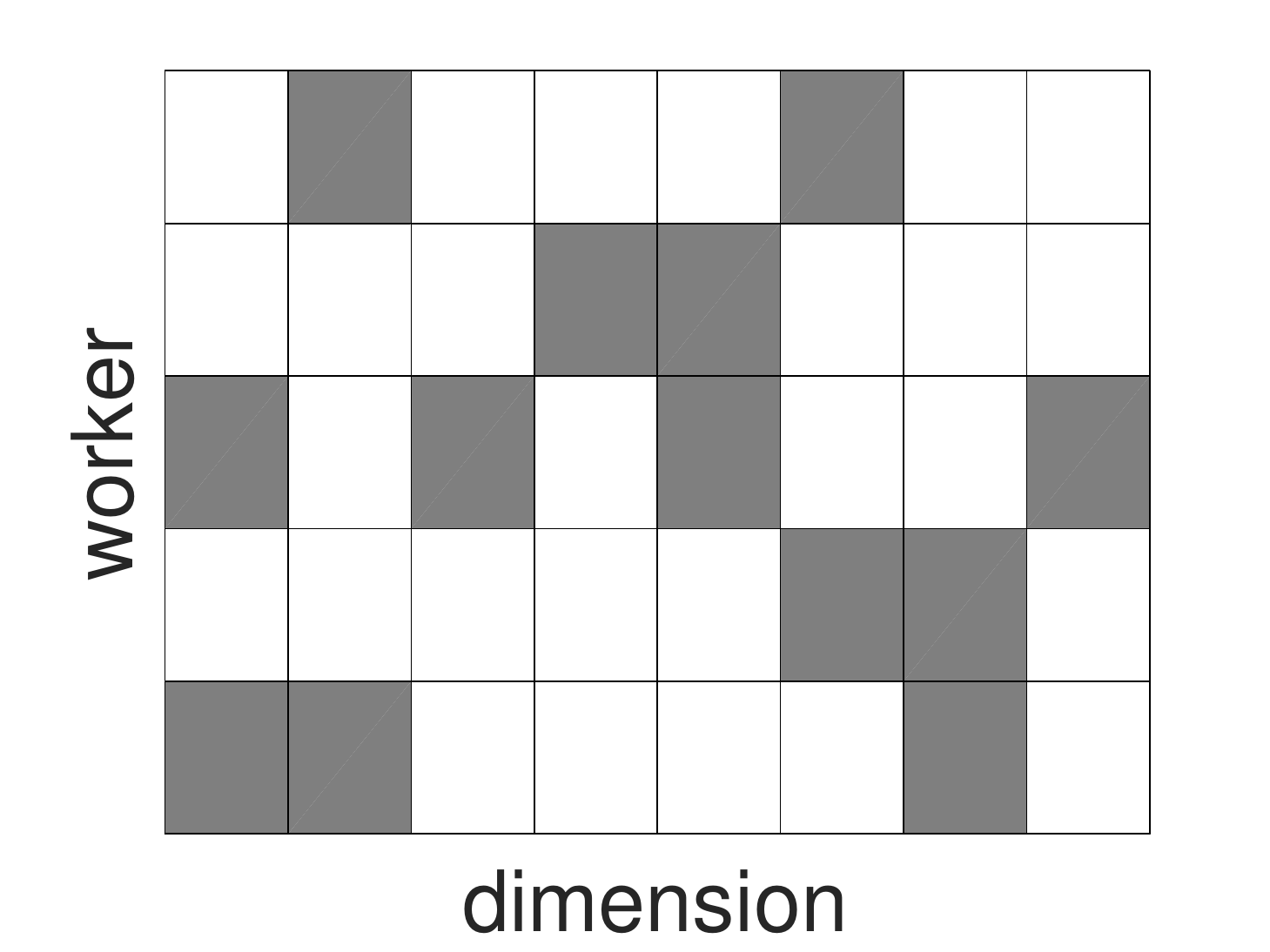}}
\caption{The 2 figures visualize $5$ workers with $8$-dimensional gradients. The $i$th row represents the gradient vector produced by the $i$th worker. The $j$th column represents the $j$th dimension of the gradients. A shadow block represents that the corresponding value is replaced by a Byzantine value. In the two examples, the maximal number of Byzantine values for each dimension is $2$. For the classic Byzantine model, all the Byzantine values must lie in the same workers~(rows), while for the generalized Byzantine model there is no such constraint. Thus, (a) is a special case of (b).}
\label{fig:viz_byz}
\end{figure}

The failure model can be described by using an $n \times d$ matrix consisting of the $d$-dimensional gradients produced by $n$ workers, as visualized in Figure~\ref{fig:viz_byz}. A previous work~\cite{blanchard2017machine} discusses a special case of our failure model, where the Byzantine values must lie in the same rows~(workers) as shown in Figure~\ref{fig:viz_byz}(a). Our failure model generalize the classic Byzantine failure model by placing the Byzantine values anywhere in the matrix without any constraint.

%To be more specific, we propose high-dimensional median as robust aggregation rules on the server side, which tolerate a subset of Byzantine workers (workers that behaves arbitrary). Such aggregation rules produce robust estimators of the true gradients given $n$ noisy gradients in $\R^d$. 

%In the Byzantine failure model, any data in transmission can be replaced by arbitrary values. Without the loss of generality, to simplify the failure model, we only consider that the gradients sent from the workers to the servers are replaced by Byzantine values. The Byzantine values can result from direct attacks on the gradient transmission, or from attacks on the broadcast parameters from servers to workers. 

There are many possible types of attacks. In general, the attackers want to disturb the model training, i.e., make SGD converge slowly or converge to a bad solution. We list some of the possible attacks in the following three paragraphs.

We name the most general type of attacks as \textit{gamber}. The attackers can change a portion of data on the communication media such as the wires or the network interfaces. The attackers randomly pick the data and maliciously change them~(e.g., multiply them by a large negative value). As a result, on the server nodes, the collected gradients are partially replaced by arbitrary values. 

Another possible type of attack is called \textit{omniscient}. The attackers are supposed to know the gradients sent by all the workers, and use the sum of all the gradients, scaled by a large negative value, to replace some of the gradient vectors. The goal is to mislead SGD to go into an opposite direction with a large step size. 

There are also some weaker attacks, such as \textit{Gaussian attack}, where some of the gradient vectors are replaced by random vectors sampled from a Gaussian distribution with large variances. Such attackers do not require any information from the workers.

With the generalized Byzantine failure model, we ask that using what aggregation rules and on what conditions, the synchronous SGD can still converge to good solutions. 
We propose novel median-based aggregation rules, with which SGD is Byzantine resilient on a certain condition: for each dimension, in all the $n$ values provided by the $n$ workers, the number of Byzantine values must be less than half of $n$. Such Byzantine resilience property is called ``dimensional Byzantine resilience".
The main contributions of this paper are listed below: 
\setitemize[0]{leftmargin=*}
\begin{itemize} 
\item We propose three aggregation rules for synchronous SGD with provable convergence to critical points: geometric median~(Definition~\ref{def:geomed}), marginal median~(Definition~\ref{def:marmed}), and ``mean around median"~(Definition~\ref{def:meamed}). As far as we know, this paper is the first to theoretically and empirically study median-based aggregation rules under non-convex settings.
\item We show that the three proposed robust aggregation rules have low computation cost. The time complexities are nearly linear, which are in the same order of the default choice for non-Byzantine aggregation, i.e., averaging.
\item We formulate the dimensional Byzantine resilience property, and prove that marginal median and ``mean around median" are dimensional Byzantine-resilient~(Definition~\ref{def:dim_byz}). As far as we know, this paper is the first one to study generalized Byzantine failures and dimensional Byzantine resilience for synchronous SGD.
\end{itemize}

\section{Model}

We consider the Parameter Server architecture consisting of $n$ workers. The goal is to find the optimizer of the following problem:
\begin{align*}
\min_x \E \left[ f(x, \xi) \right],
\end{align*}
where the expectation is with respect to the random variable $\xi$.
The PS executes synchronous SGD for distributed training. In each round, the server nodes collect $n$ gradients from the workers.  
%For each dimension of the parameters, up to $q$ of the $n$ values are Byzantine~(behaving arbitrarily). 
In the $t^{\text{th}}$ round, the server nodes aggregate the gradients $\{\tilde{v}_i^t: i \in [n]\}$ from the workers, and broadcast the updated parameters $x^{t+1}$ to the workers. $\tilde{v}_i^t$ is the vector sent by the $i$th worker in the $t$th round, potentially Byzantine.  Using aggregation rule $Aggr(\cdot)$, the server nodes update the parameters as follows:
\[
x^{t+1} \leftarrow x^t - \gamma^t Aggr(\{\tilde{v}_i^t: i \in [n]\}),
\]
where $\gamma^t$ is the learning rate. The worker nodes pull the latest parameters from the server nodes, compute the gradients according to the local portion of the training data, and send the gradients to the server nodes. Without the Byzantine failures, the $i$th worker will calculate $v_i^t \sim G^t$, where $G^t = \nabla f(x^t, \xi)$. With Byzantine failures, $v_i^t$ are partially replaced by any arbitrary values, which results $\tilde{v}_i^t$. 

Since the Byzantine failure assumes the worst cases, the attackers may have full knowledge of the entire system, including the gradients generated by all the workers, and the aggregation rule $Aggr(\cdot)$. The malicious processes can even collaborate with each other~\cite{lynch1996distributed}.

\section{Byzantine Resilience}

In this section, we formally define the classic Byzantine resilience property and its generalized version: dimensional Byzantine resilience.

Suppose that in a specific round, the correct vectors $\{v_i: i \in [n]\}$ are i.i.d samples drawn from the random variable $G = \nabla f(x, \xi)$, where $\E[G] = g$ is an unbiased estimator of the gradient. Thus, $\E[v_i] = \E[G] = g$, for any $i \in [n]$. We simplify the notations by ignoring the index of round $t$.
%, which is same for all the symbols in this section. 

We first introduce the classic Byzantine model proposed by \citet{blanchard2017machine}.
With the Byzantine workers, the actual vectors $\{\tilde{v}_i: i \in [n]\}$ received by the server nodes are as follows: 
\begin{definition}[Classic Byzantine Model]
\begin{align}
\tilde{v}_i = 
\begin{cases}
v_i, \mbox{if the $i$th worker is correct,}\\
arbitrary, \mbox{if the $i$th worker is Byzantine}.
\end{cases}
\label{equ:byz_worker}
\end{align}
\end{definition}
Note that the indices of Byzantine workers can change throughout different rounds. Furthermore, the server nodes are not aware of which workers are Byzantine. The only information given is the number of Byzantine workers, if necessary.

We directly use the same definition of classic Byzantine resilience proposed in \cite{blanchard2017machine}.
\begin{definition} 
\label{def:byz}
(Classic $(\alpha,q)$-Byzantine Resilience). Let $0 \leq \alpha < \pi/2$ be any angular value, and any integer $0 \leq q \leq n$. Let $\{v_i: i \in [n]\}$ be any i.i.d. random vectors in $\R^d$, $v_i \sim G$, with $\E[G] = g$. Let $\{\tilde{v}_i: i \in [n]\}$ be the set of vectors, of which up to $q$ of them are replaced by arbitrary vectors in $\R^d$, while the others still equal to the corresponding $\{v_i\}$. Aggregation rule $Aggr(\cdot)$ is said to be classic $(\alpha, q)$-Byzantine resilient if 
$
Aggr(\{\tilde{v}_i: i \in [n]\})
$
satisfies (i) $\ip{\E[Aggr]}{g} \geq (1-\sin \alpha) \|g\|^2 > 0$ and (ii) for $r = 2,3,4$, $\E\|Aggr\|^r$ is bounded above by a linear combination of terms $\E\|G\|^{r_1}, \ldots, \E\|G\|^{r_{n-q}}$ with $r_1 + \ldots + r_{n-q} = r$.
\end{definition}

The baseline algorithm \textit{Krum}, denoted as $Krum(\{\tilde{v}_i: i \in [n]\})$~\cite{blanchard2017machine}, is defined as follows
\begin{definition} 
\label{def:krum}
\begin{align*}
&Krum(\{\tilde{v}_i: i \in [n]\}) = \tilde{v}_k, \\
&k = \argmin_{i \in [n]} \sum_{i \rightarrow j} \| \tilde{v}_i - \tilde{v}_j \|^2,
\end{align*}
where $i \rightarrow j$ is the indices of the $n-q-2$ nearest neighbours of $\tilde{v}_i$ in $\{\tilde{v}_i: i \in [n]\}$ measured by Euclidean distance.
\end{definition}
The \textit{Krum} aggregation is classic $(\alpha,q)$-Byzantine resilient under certain assumptions:
\begin{lemma}[\citet{blanchard2017machine}]
Let $v_1, \ldots, v_n$ be any i.i.d. random $d$-dimensional vectors s.t. $v_i \sim G$, with $\E[G] = g$ and $\E\|G-g\|^2 = d\sigma^2$. $q$ of $\{v_i: i \in [n]\}$ are replaced by arbitrary $d$-dimensional vectors $b_1, \ldots, b_q$. If $2q + 2 < n$ and $\eta_0(n, q)\sqrt{d}\sigma < \|g\|$, where 
\[
\eta_0^2(n, q) = 2\left(n-q + \frac{q(n-q-2) + q^2(n-q-1)}{n-2q-2} \right),
\]
then the $Krum$ function is classic $(\alpha_0, q)$-Byzantine resilient where $0 \leq \alpha_0 < \pi/2$ is defined by 
$\sin \alpha_0 = \frac{\eta_0(n,q) \sqrt{d} \sigma}{\|g\|}$.
\end{lemma}

The generalized Byzantine model is denoted as:
\begin{definition}[Generalized Byzantine Model]
\begin{align}
(\tilde{v}_i)_j = 
\begin{cases}
(v_i)_j, \mbox{if the the $j$th dimension of $v_i$ is correct,}\\
arbitrary, \mbox{otherwise},
\end{cases}
\label{equ:byz_model}
\end{align}
where $(v_i)_j$ is the $j$th dimension of the vector $v_i$.
\end{definition}

Based on the Byzantine model above, we introduce a generalized Byzantine resilience property, dimensional $(\alpha,q)$-Byzantine resilience, which is defined as follows:
\begin{definition} 
\label{def:dim_byz}
(Dimensional $(\alpha,q)$-Byzantine Resilience). Let $0 \leq \alpha < \pi/2$ be any angular value, and any integer $0 \leq q \leq n$. Let $\{v_i: i \in [n]\}$ be any i.i.d. random vectors in $\R^d$, $v_i \sim G$, with $\E[G] = g$. Let $\{\tilde{v}_i: i \in [n]\}$ be the set of vectors. For each dimension, up to $q$ of the $n$ values are replaced by arbitrary values, i.e., for dimension $j \in [d]$, $q$ of $\{(\tilde{v}_i)_j: i \in [n]\}$ are Byzantine, where $(\tilde{v}_i)_j$ is the $j$th dimension of the vector $\tilde{v}_i$. Aggregation rule $Aggr(\cdot)$ is said to be dimensional $(\alpha, q)$-Byzantine resilient if 
$
Aggr(\{\tilde{v}_i: i \in [n]\})
$
satisfies (i) $\ip{\E[Aggr]}{g} \geq (1-\sin \alpha) \|g\|^2 > 0$ and (ii) for $r = 2,3,4$, $\E\|Aggr\|^r$ is bounded above by a linear combination of terms $\E\|G\|^{r_1}, \ldots, \E\|G\|^{r_{n-q}}$ with $r_1 + \ldots + r_{n-q} = r$.
\end{definition}

Note that classic $(\alpha,q)$-Byzantine resilience is a special case of dimensional $(\alpha,q)$-Byzantine resilience. For classic Byzantine resilience defined in Definition~\ref{def:byz}, all the Byzantine values must lie in the same subset of workers, as shown in Figure~\ref{fig:viz_byz}(a).

In the following theorems, we show that \textit{Mean} and \textit{Krum} are not dimensional Byzantine resilient. The proofs are provided in the appendix.
\begin{theorem}
\label{thm:mean_dim_byz}
Averaging is not dimensional Byzantine resilient.
\end{theorem}

\begin{theorem}
\label{thm:dim_byz}
Any aggregation rule $Aggr(\{\tilde{v}_i: i \in [n]\})$ that outputs $Aggr \in \{\tilde{v}_i: i \in [n]\}$ is not dimensional Byzantine resilient.
\end{theorem}
Note that \textit{Krum} chooses the vector $v \in \{\tilde{v}_i: i \in [n]\}$ with the minimal score. Thus, based on the theorem above, we obtain the following corollary.
\begin{corollary}
\label{lem:krum_dim_byz}
$Krum(\cdot)$ is not dimensional Byzantine resilient.
\end{corollary}
%\begin{proof}
%We only need to demonstrate one example in which $Krum$ fails.
%Consider the case where the $i$th dimension of the $i$th vector $v_i$ is manipulated by the malicious workers (e.g. replaced by an arbitrarily large value or with flipped sign), where $i \in [n]$. Thus, up to 1 value of each dimension is Byzantine. For Krum, all the $n$ vectors are manipulated, which  break the assumption that $2q+2 < n$. Thus, $Krum$ can not even tolerate one single Byzantine value~($q=1$) of each dimension, if such Byzantine values are placed in different workers for different dimensions.
%\end{proof}

If an aggregation rule is dimensional/classic $(\alpha,q)$-Byzantine resilient with satisfied assumptions, it converges to critical points almost surely, by reusing the Proposition 2 in \cite{blanchard2017machine}. We provide the following lemma without proof.
\begin{lemma}[\citet{blanchard2017machine}]
Assume that (i) the cost function $f$ is three times differentiable with continuous derivatives, and is non-negative,  $f(x) \geq 0$; (ii) the learning rates satisfy $\sum_t \gamma_t = \infty$ and $\sum_t \gamma^2_t < \infty$; (iii) the gradient estimator satisfies $\E[\nabla f(x, \xi)] = \nabla F (x)$ and $\forall r \in \{2,3,4\}$, $\E\|\nabla f(x,\xi)\|^r \leq A_r + B_r\|x\|^r$ for some constants $A_r$, $B_r$; (iv) there exists a constant $0 \leq \alpha < \pi/2$ such that for all x 
$
\eta(n, q) \sqrt{d} \sigma(x) \leq \|\nabla F(x)\| \sin\alpha,
$
where $d \sigma^2(x) = \E \|\nabla f(x, \xi) - \nabla F(x)\|^2$; (v) finally, beyond a certain horizon, $\|x\|^2 \geq D$, there exist $\epsilon > 0$ and $0 \leq \beta < \pi/2 - \alpha$ such that $\|\nabla F(x)\| \geq \epsilon > 0$, and $\frac{\ip{x}{\nabla F(x)}}{\|x\| \cdot \|\nabla F(x)\|} \geq \cos\beta$.
%\begin{align*}
%\|\nabla F(x)\| &\geq \epsilon > 0 \\
%\frac{\ip{x}{\nabla F(x)}}{\|x\| \cdot \|\nabla F(x)\|} &\geq \cos\beta.
%\end{align*}
Then the sequence of gradients $\nabla F(x^t)$ converges almost surely to zero, if the aggregation rule satisfies $(\alpha, q)$-Byzantine Resilience defined in Definition~\ref{def:byz} or \ref{def:dim_byz}.
\end{lemma}

\section{Median-based Aggregation}
\label{sec:median}
%There are two major disadvantages of $Krum$:
%\setitemize[0]{leftmargin=*}
%\begin{itemize}
%\item The computation is expensive. The time complexity of $Krum$ is $O(n^2 d)$, while averaging only takes $O(nd)$.
%\item $Krum$ is not dimensional Byzantine resilient as shown in Lemma~\ref{lem:krum_dim_byz}.
%\end{itemize}
%We wonder whether there are Byzantine resilient aggregation rules with time complexity nearly linear to $n$. Furthermore, are there dimensional Byzantine resilient aggregation rules?
%It turns out that the solution lies in median-based aggregation rules. 
%In this section, we define two high-dimensional medians, which generalize one-dimensional median, and and one variant. Furthermore, we show their Byzantine resilience properties.

With the Byzantine failure model defined in Equation~(\ref{equ:byz_worker}) and (\ref{equ:byz_model}), we propose three median-based aggregation rules, which are Byzantine resilient under certain conditions. 

\subsection{Geometric Median}
The geometric median is used as a robust estimator of mean~\citep{chen2017distributed}. 
\begin{definition}
\label{def:geomed}
The geometric median of $\{\tilde{v}_i: i \in [n]\}$, denoted by $GeoMed(\{\tilde{v}_i: i \in [n]\})$, is defined as 
\begin{align*}
\lambda = GeoMed(\{\tilde{v}_i: i \in [n]\}) = \argmin_{v \in \R^d} \sum_{i=1}^n \|v - \tilde{v}_i\|.
\end{align*}
\end{definition}

The following theorem shows the classic $(\alpha_1, q)$-Byzantine resilience of geometric median. A proof is provided in the appendix.

\begin{theorem}
Let $v_1, \ldots, v_n$ be any i.i.d. random $d$-dimensional vectors s.t. $v_i \sim G$, with $\E[G] = g$ and $\E\|G-g\|^2 = d\sigma^2$. $q$ of $\{v_i: i \in [n]\}$ are replaced by arbitrary $d$-dimensional vectors $b_1, \ldots, b_q$. If $q \leq \ceil{\frac{n}{2}}-1$ and $\eta_1(n, q)\sqrt{d}\sigma < \|g\|$, where $\eta_1(n, q) = \frac{2n-2q}{n-2q} \sqrt{n-q}$, then the $GeoMed$ function is classic $(\alpha_1, q)$-Byzantine resilient where $0 \leq \alpha_1 < \pi/2$ is defined by 
$\sin \alpha_1 = \frac{\eta_1(n,q) \sqrt{d} \sigma}{\|g\|}$.
\end{theorem}

\subsection{Marginal Median}
The marginal median is another generalization of one-dimensional median. 
\begin{definition}
\label{def:marmed}
We define the marginal median aggregation rule $MarMed(\cdot)$ as
\begin{align*}
\mu = MarMed(\{\tilde{v}_i: i \in [n]\}),
\end{align*}
where for any $j\in[d]$, the $j$th dimension of $\mu$ is $\mu_j = median\left(\{(\tilde{v}_1)_j, \ldots, (\tilde{v}_n)_j\}\right)$, $(\tilde{v}_i)_j$ is the $j$th dimension of the vector $\tilde{v}_i$, $median(\cdot)$ is the one-dimensional median.  
\end{definition}

The following theorem claims that by using $MarMed(\cdot)$, the resulting vector is dimensional $(\alpha_2, q)$-Byzantine resilient. A proof is provided in the appendix.

\begin{theorem}
Let $v_1, \ldots, v_n$ be any i.i.d. random $d$-dimensional vectors s.t. $v_i \sim G$, with $\E[G] = g$ and $\E\|G-g\|^2 = d\sigma^2$. For any dimension $j \in [d]$, $q$ of $\{(v_1)_j, \dots, (v_n)_j\}$ are replaced by arbitrary values, where $(v_i)_j$ is the $j$th dimension of the vector $v_i$. If $q \leq \ceil{\frac{n}{2}}-1$ and $\eta_2(n, q)\sqrt{d}\sigma < \|g\|$, where $\eta_2(n, q) = \sqrt{n-q}$, then the $MarMed$ function is dimensional $(\alpha_2, q)$-Byzantine resilient where $0 \leq \alpha_2 < \pi/2$ is defined by 
$\sin \alpha_2 = \frac{\eta_2(n,q) \sqrt{d} \sigma}{\|g\|}$.
\end{theorem}

%\begin{remark}
%Note that in \cite{blanchard2017machine}, $Krum$ satisfies, 
%\[
%\eta^2(n, q) = 2\left(n-q + \underbrace{\frac{q(n-q-2) + q^2(n-q-1)}{n-2q-2}}_{\mbox{positive}} \right),
%\]
%while in our bound, $\eta_2^2(n, q) = n-q$, which reduces more than half of the upper bound of $\|\E[Krum] - g\|^2$.
%\end{remark}

\subsection{Beyond Median}
We can also utilize more values for each dimension along with the median, if $q$ is given or easily estimated.
To be more specific, for each dimension, we take the average of the $n-q$ values nearest to the median (including the median itself). We call the resulting aggregation rule \textit{``mean around median"}, which is defined as follows:
\begin{definition}
\label{def:meamed}
We define the mean-around-median aggregation rule $MeaMed(\cdot)$ as
\begin{align*}
\rho = MeaMed(\{\tilde{v}_i: i \in [n]\}),
\end{align*}
where for any $j\in[d]$, the $j$th dimension of $\rho$ is $\rho_j = \frac{1}{n-q} \sum_{\mu_j \rightarrow i} (\tilde{v}_i)_j$, $\mu_j \rightarrow i$ is the indices of the top-$(n-q)$ values lying in $\{(\tilde{v}_1)_j, \ldots, (\tilde{v}_n)_j\}$ nearest to the median $\mu_j$, $(\tilde{v}_i)_j$ is the $j$th dimension of the vector $\tilde{v}_i$.  
\end{definition}

We show that $MeaMed$ is dimensional $(\alpha_3, q)$-Byzantine resilient.
\begin{theorem}
Let $v_1, \ldots, v_n$ be any i.i.d. random $d$-dimensional vectors s.t. $v_i \sim G$, with $\E[G] = g$ and $\E\|G-g\|^2 = d\sigma^2$. For any dimension $j \in [d]$, $q$ of $\{(v_1)_j, \dots, (v_n)_j\}$ are replaced by arbitrary values, where $(v_i)_j$ is the $j$th dimension of the vector $v_i$. If $q \leq \ceil{\frac{n}{2}}-1$ and $\eta_3(n, q)\sqrt{d}\sigma < \|g\|$, where $\eta_3(n, q) = \sqrt{10(n-q)}$, then the $MeaMed$ function is dimensional $(\alpha_3, q)$-Byzantine resilient where $0 \leq \alpha_3 < \pi/2$ is defined by 
$\sin \alpha_3 = \frac{\eta_3(n,q) \sqrt{d} \sigma}{\|g\|}$.
\end{theorem}

The mean-around-median aggregation can be viewed as a trimmed average centering at the median, which filters out the values far away from the median.

\subsection{Time Complexity}
For geometric median $GeoMed(\cdot)$, there are no closed-form solutions. The $(1+\epsilon)$-approximate geometric median can be computed in $O(dn \log^3 \frac{1}{\epsilon})$ time~\citep{cohen2016geometric}, which is nearly linear to $O(dn)$. 
To compute the marginal median $MarMed(\cdot)$, we only need to compute the median value of each dimension. The simplest way is to apply any sorting algorithm to each dimension, which yields the time complexity $O(dn\log n)$. To obtian median values, there also exists an algorithm called \textit{selection algorithm}~\cite{blum1973time} with average time complexity $O(n)$~($O(n^2)$ in the worst case). Thus, we can get the marginal median with time complexity $O(dn)$ on average, which is in the same order of using mean value for aggregation. For $MeaMed(\cdot)$, the computation additional to computing the marginal median takes linear time $O(dn)$. Thus, the time complexity is the same as $MarMed(\cdot)$.
Note that for Krum and Multi-Krum, the time complexity is $O(dn^2)$~\cite{blanchard2017machine}.

\begin{table*}[htb!]
\caption{Experiment Summary}
\label{table:datasets}
\begin{center}
\begin{tabular}{|l|r|r|r|r|r|r|r|r|}
\hline 
Dataset & \# train & \# test  & $\gamma$ & \# rounds & Batchsize & Evaluation metric   \\ \hline 
MNIST~\cite{loosli2007training} & 60k & 10k & 0.1 & 500 & 32 & top-1 accuracy \\ \hline 
CIFAR10~\cite{krizhevsky2009learning} & 50k & 10k & 5e-4 & 4000 & 128 & top-3 accuracy \\ \hline 
\end{tabular} 
\end{center}
\end{table*}
\begin{figure*}[htb!]
\vspace*{-0.3cm}
\centering
\subfigure[MNIST without Byzantine]{\includegraphics[width=0.49\textwidth]{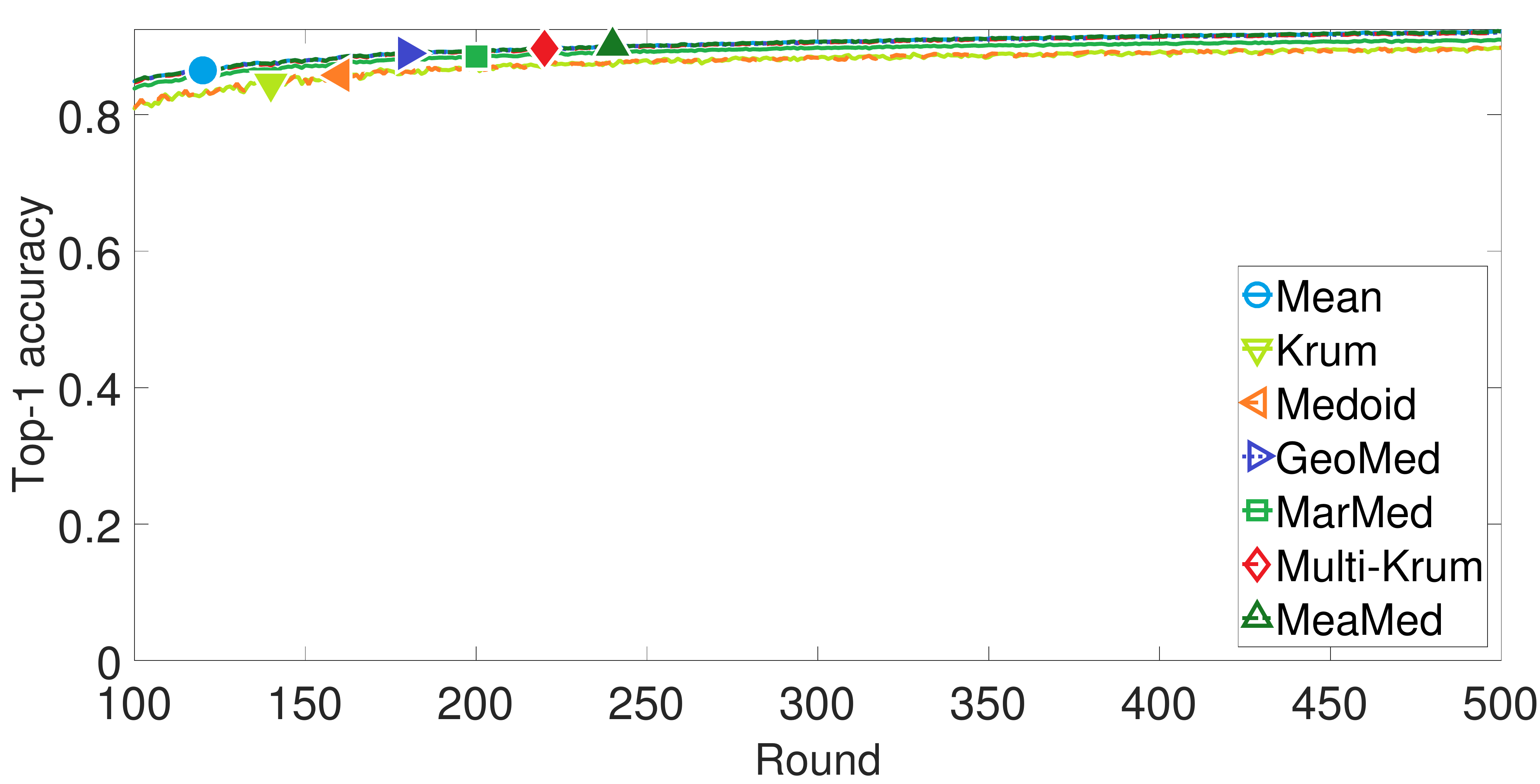}}
\subfigure[MNIST without Byzantine~(zoomed)]{\includegraphics[width=0.49\textwidth]{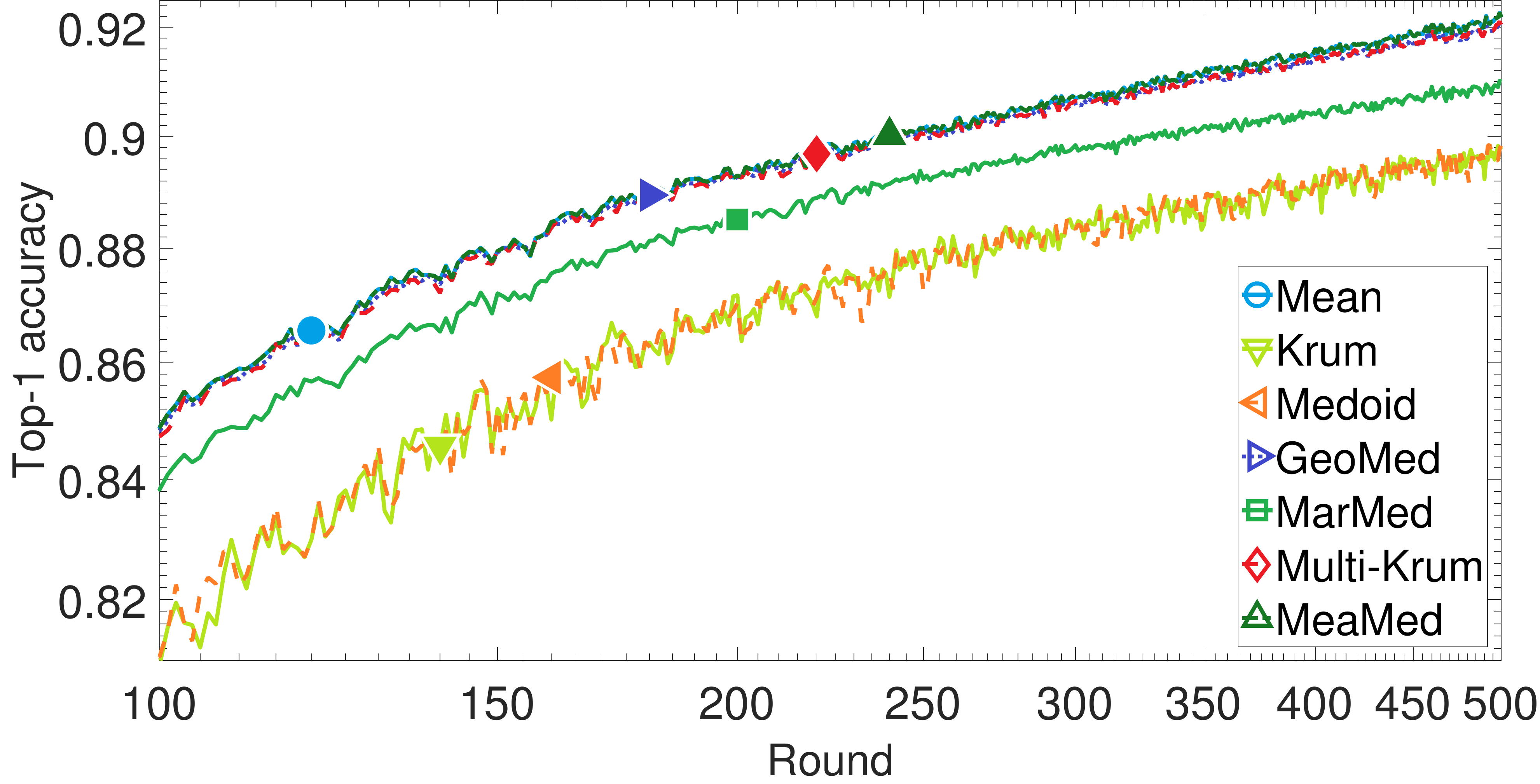}}
\caption{Top-1 accuracy of MLP on MNIST without Byzantine failures.}
\label{fig:mnist_nobyz}
\end{figure*}
\begin{figure*}[htb!]
\centering
\subfigure[MNIST with Gaussian]{\includegraphics[width=0.49\textwidth]{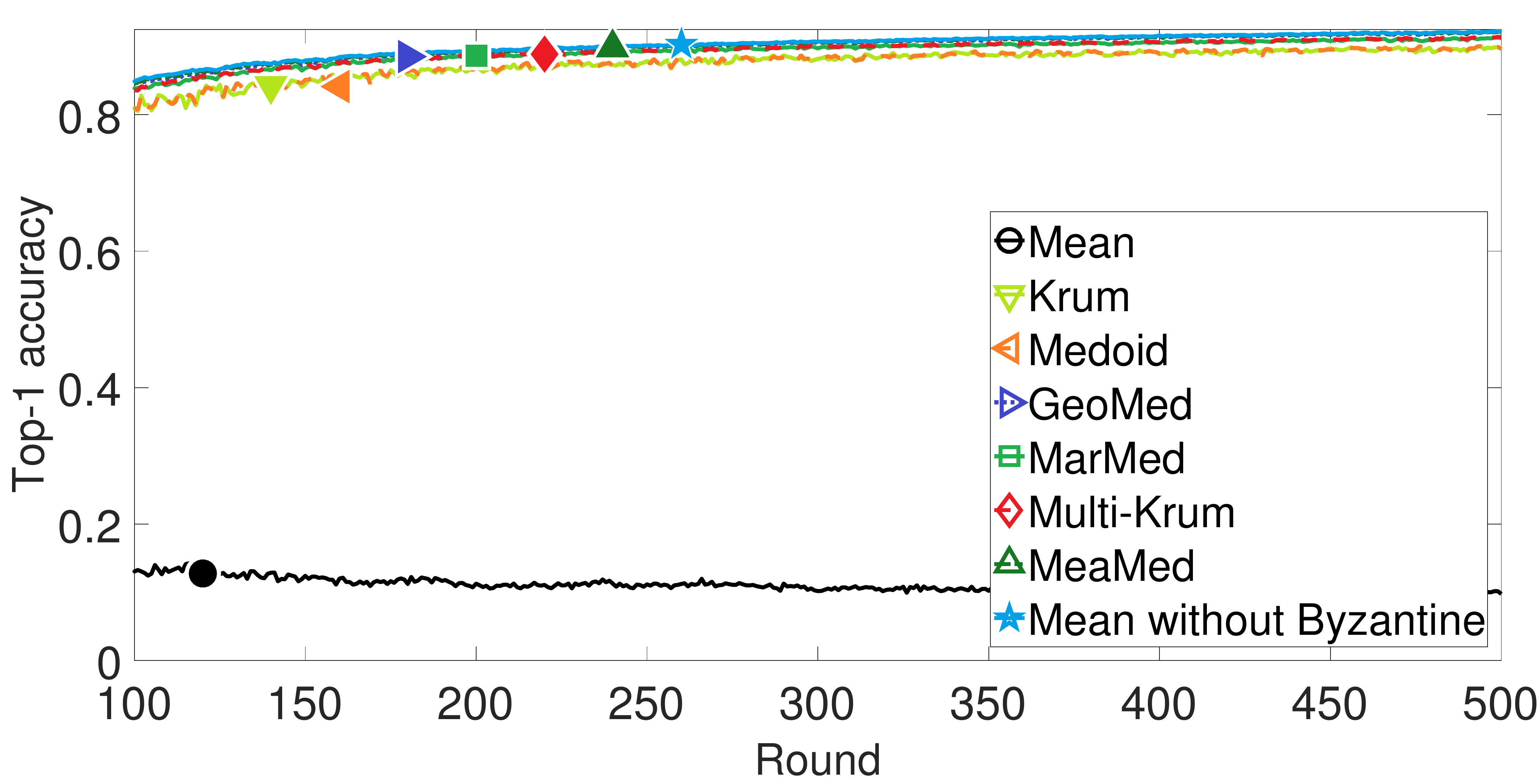}}
\subfigure[MNIST with Gaussian~(zoomed)]{\includegraphics[width=0.49\textwidth]{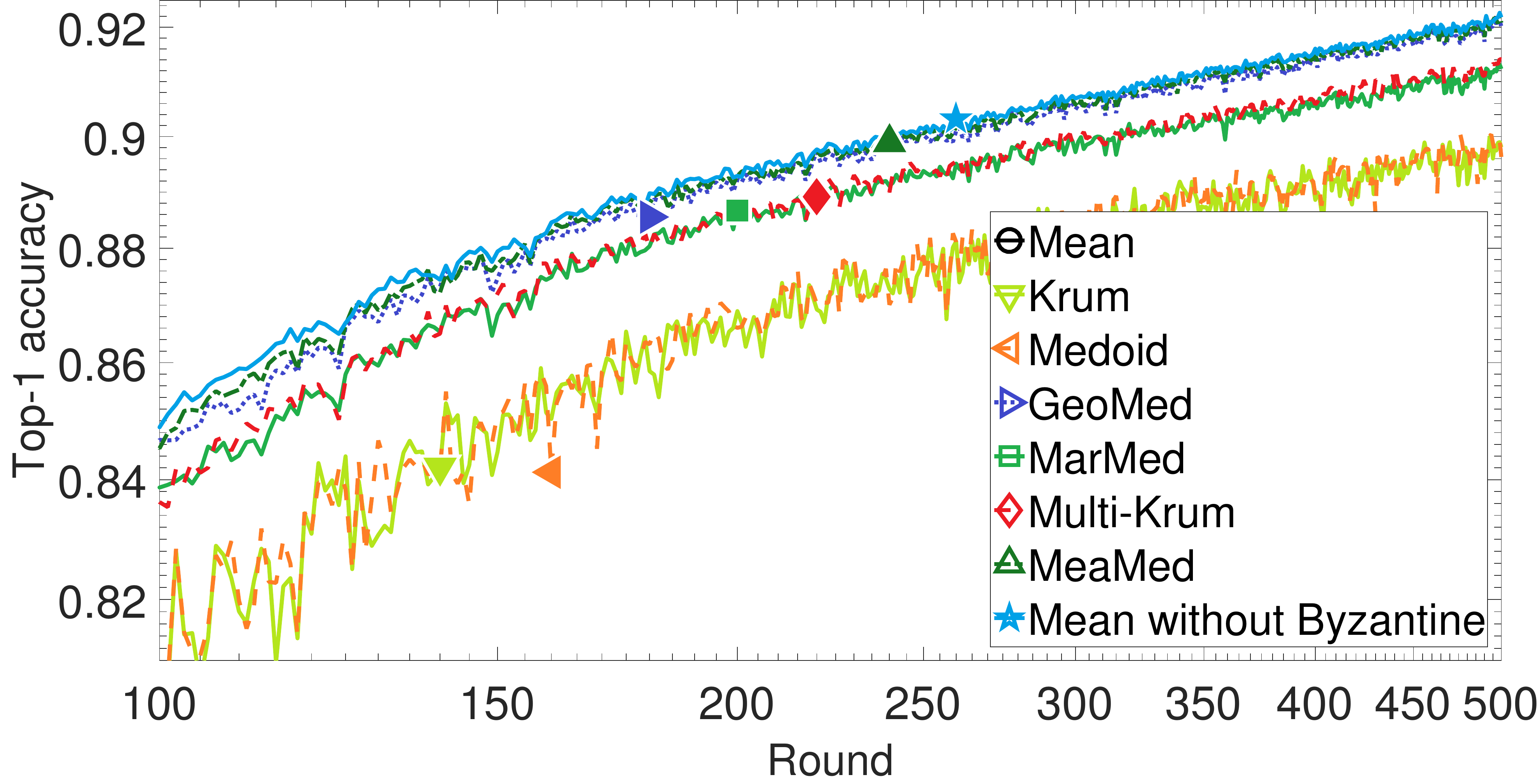}}
\caption{Top-1 accuracy of MLP on MNIST with Gaussian Attack. 6 out of 20 gradient vectors are replaced by i.i.d. random vectors drawn from a Gaussian distribution with 0 mean and 200 standard deviation.}
\label{fig:mnist_gaussian}
\end{figure*}
\begin{figure*}[htb!]
\centering
\subfigure[MNIST with omniscient]{\includegraphics[width=0.49\textwidth]{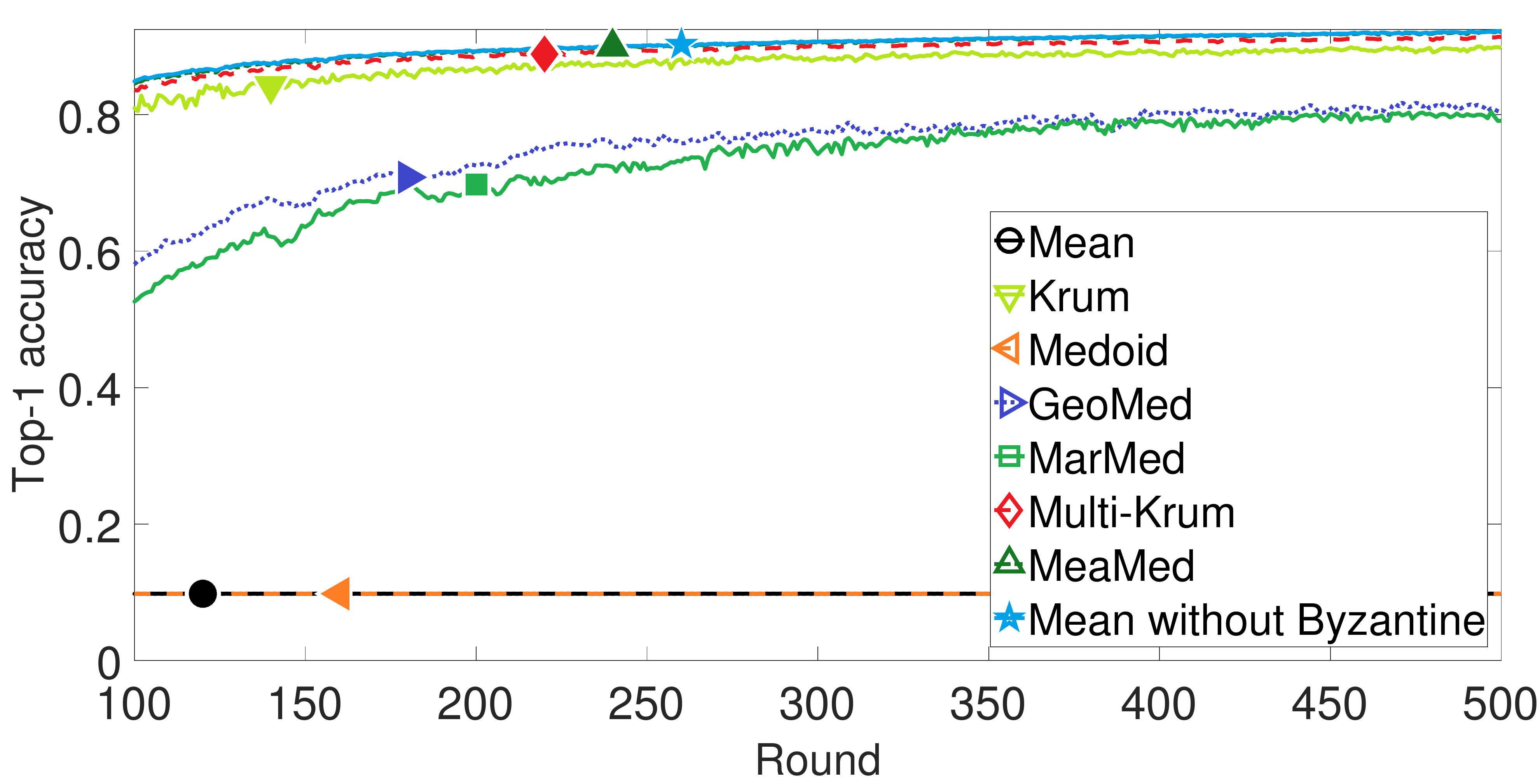}}
\subfigure[MNIST with omniscient~(zoomed)]{\includegraphics[width=0.49\textwidth]{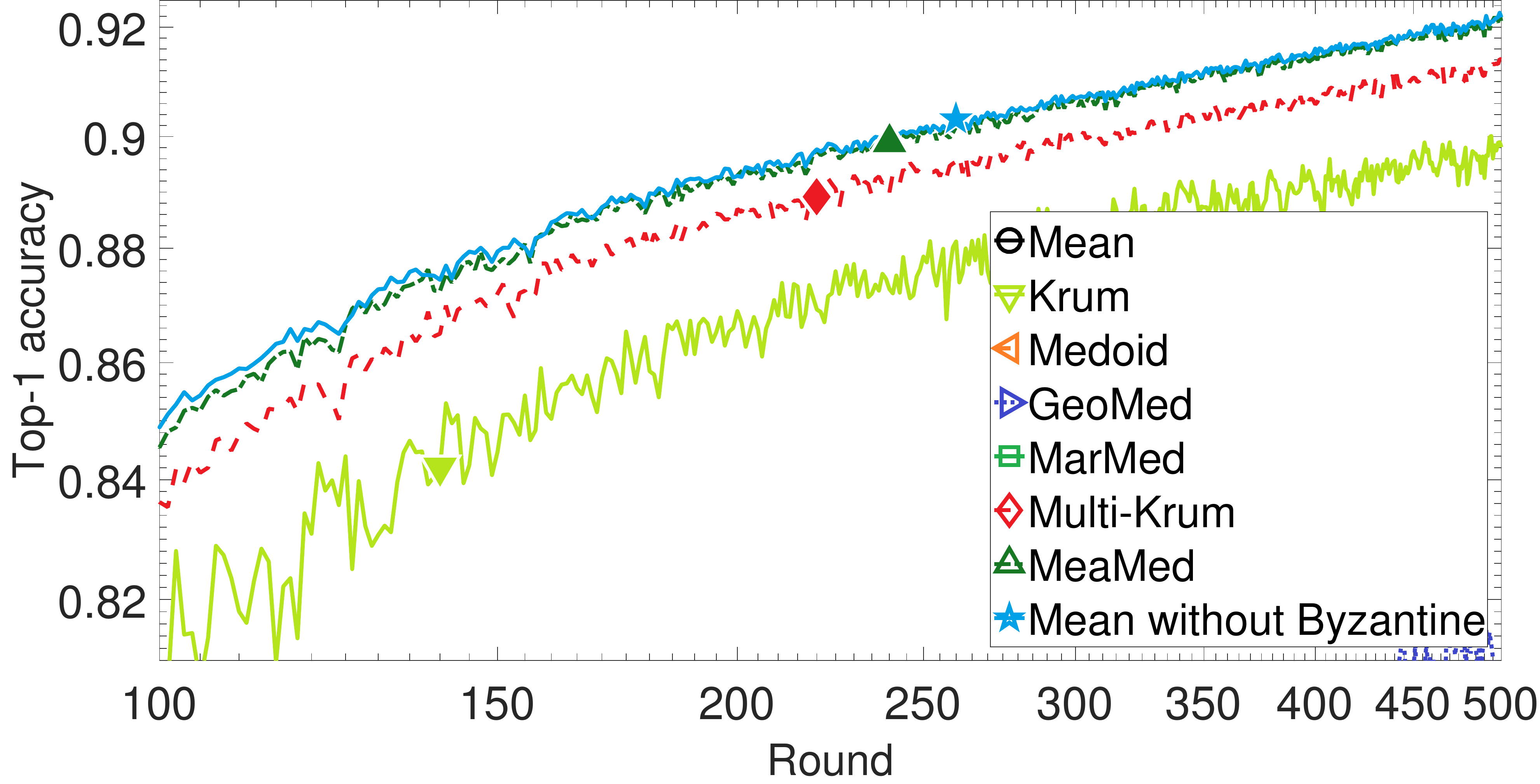}}
\caption{Top-1 accuracy of MLP on MNIST with Omniscient Attack. 6 out of 20 gradient vectors are replaced by the negative sum of all the correct gradients, scaled by a large constant~(1e20 in the experiments). }
\label{fig:mnist_omniscient}
\end{figure*}
\begin{figure*}[htb!]
\centering
\subfigure[MNIST with bit-flip]{\includegraphics[width=0.49\textwidth]{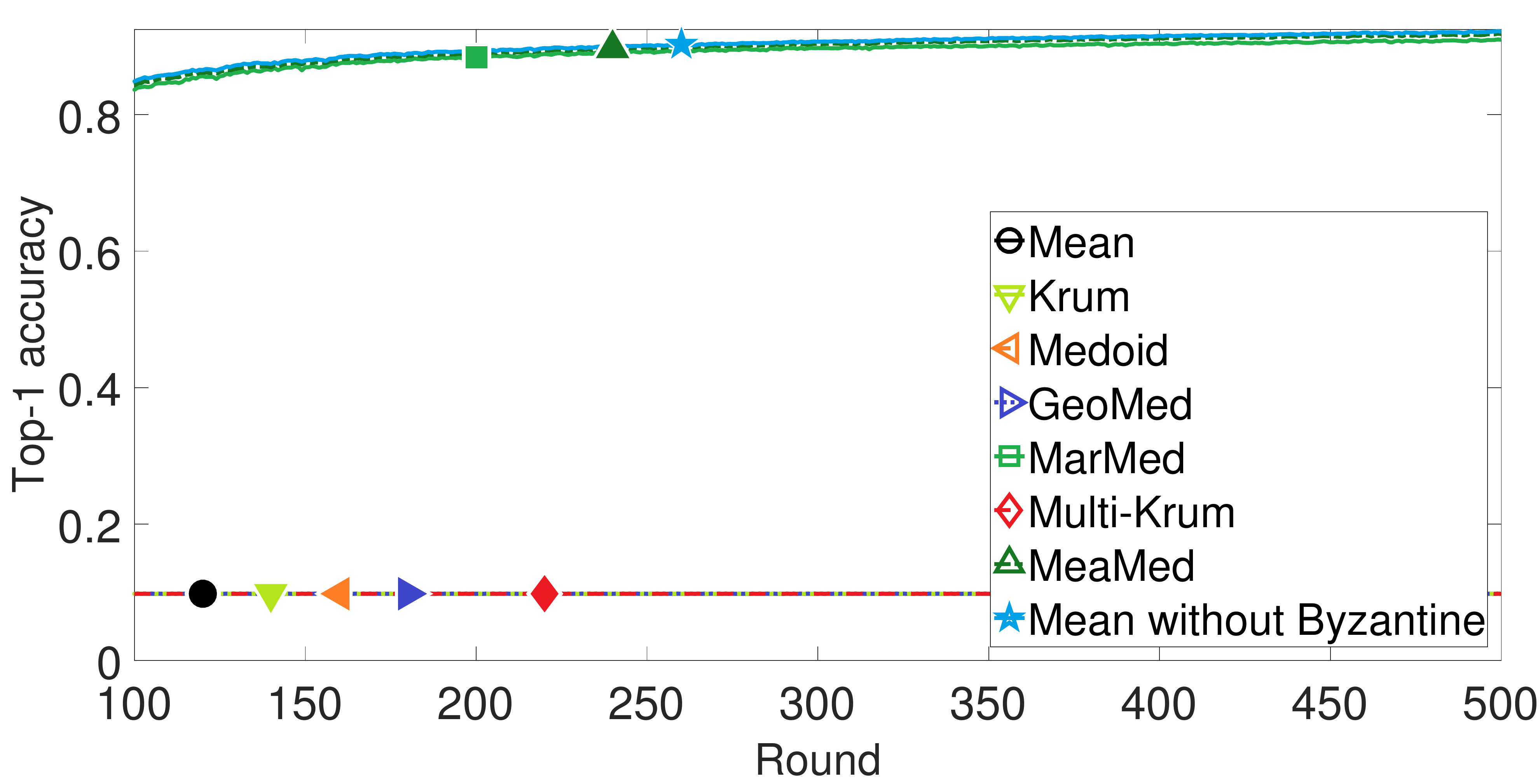}}
\subfigure[MNIST with bit-flip~(zoomed)]{\includegraphics[width=0.49\textwidth]{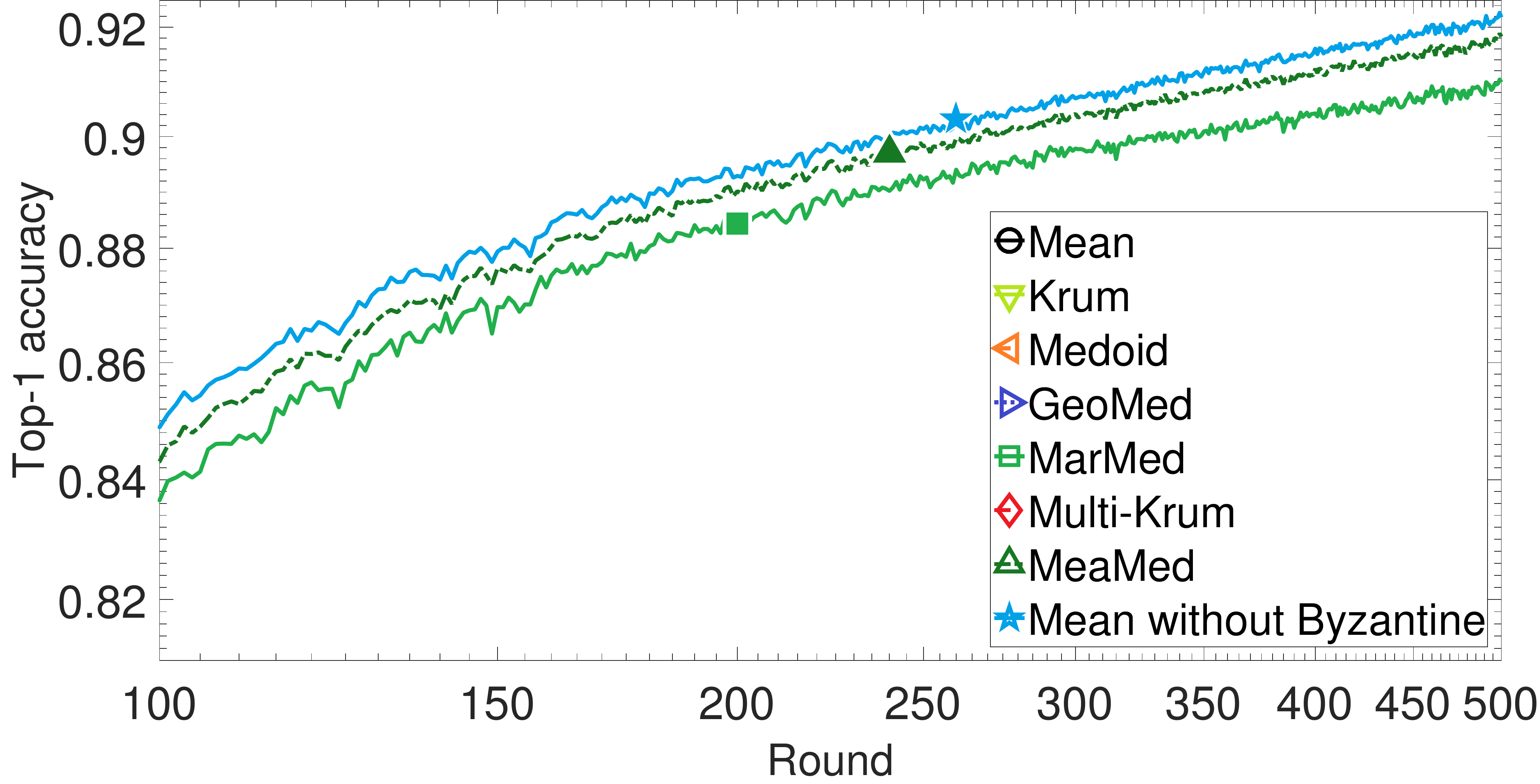}}
\caption{Top-1 accuracy of MLP on MNIST with Bit-flip Attack. For the first 1000 dimensions, 1 of the 20 floating numbers is manipulated by flipping the 22th, 30th, 31th and 32th bits.}
\label{fig:mnist_bitflip}
\end{figure*}
\begin{figure*}[htb!]
\centering
\subfigure[MNIST with gambler]{\includegraphics[width=0.49\textwidth]{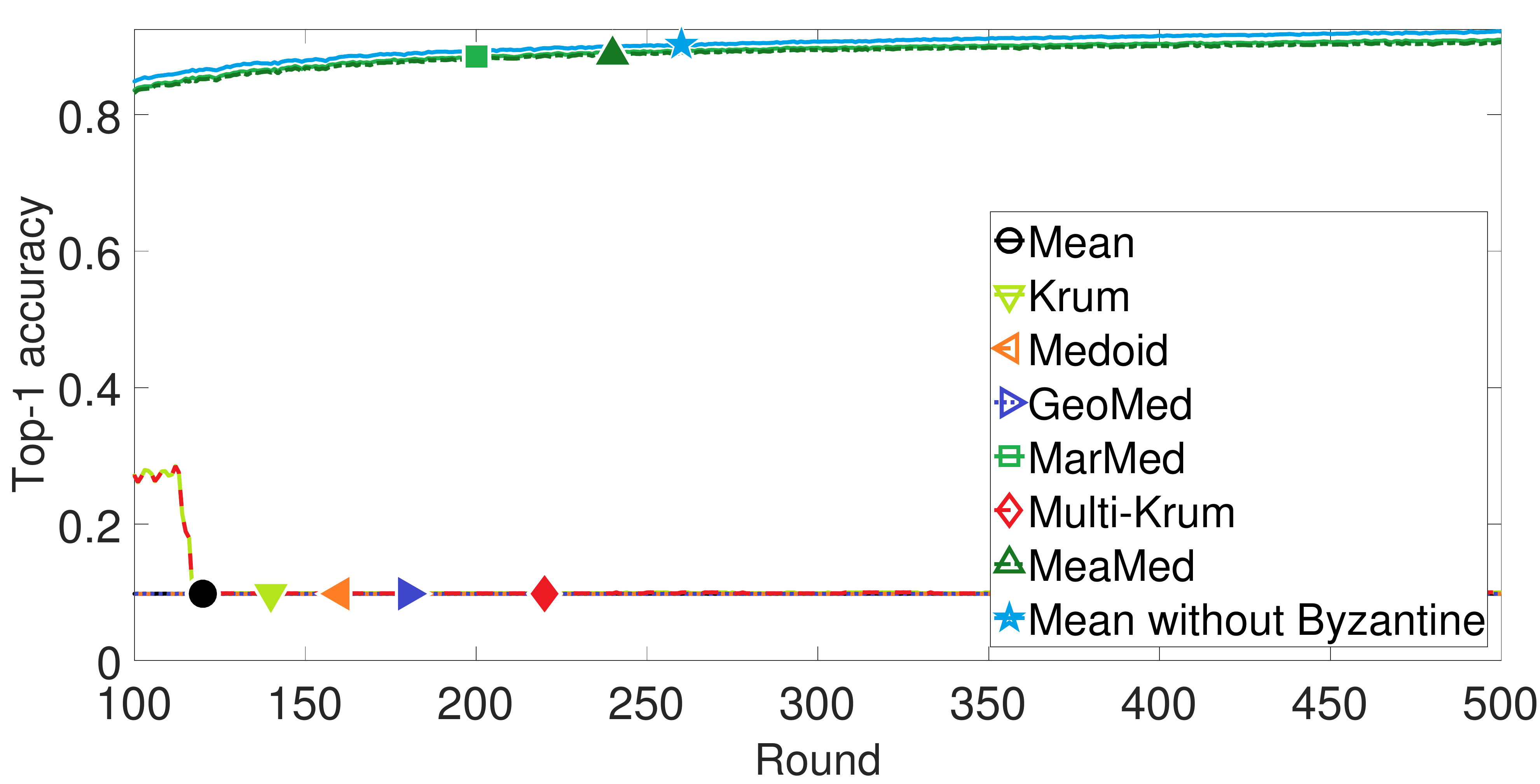}}
\subfigure[MNIST with gambler~(zoomed)]{\includegraphics[width=0.49\textwidth]{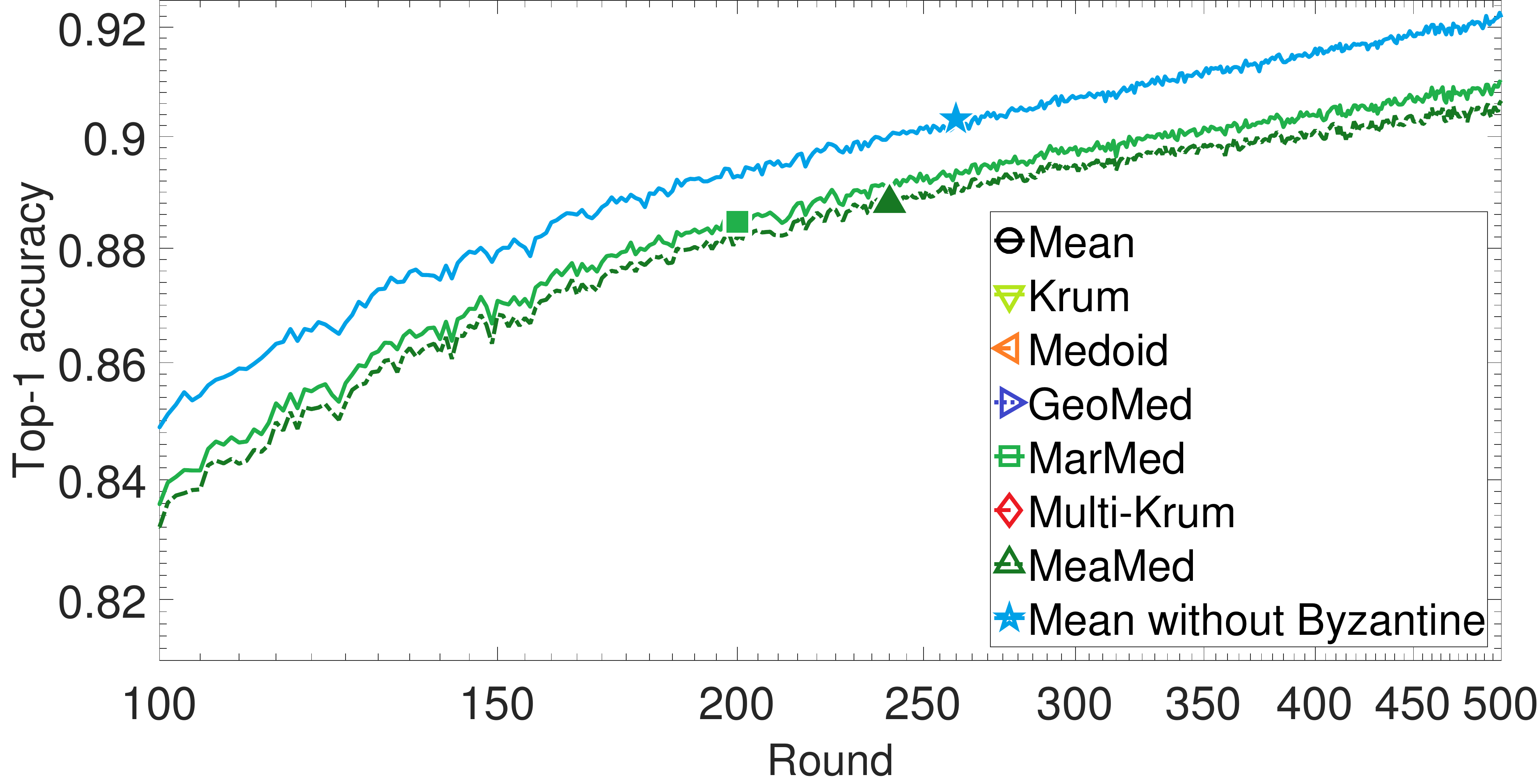}}
\caption{Top-1 accuracy of MLP on MNIST with gambler attack. The parameters are evenly assigned to 20 servers. For one single server, any received value is multiplied by $-1e20$ with probability 0.05\%.}
\label{fig:mnist_multiserver}
\end{figure*}
%\begin{figure}[htb]
%\centering
%\includegraphics[width=0.49\textwidth]{cifar10_no_byz_small}
%\caption{Top-3 Accuracy VS. \# rounds evaluated on CIFAR10 without Byzantine failures}
%\label{fig:cifar10_nobyz}
%\end{figure}
%\begin{figure}[htb]
%\centering
%\includegraphics[width=0.49\textwidth]{cifar10_gaussian_small}
%\caption{Top-3 Accuracy VS. \# rounds evaluated on CIFAR10 with Gaussian Attack}
%\label{fig:cifar10_gaussian}
%\end{figure}
%\begin{figure}[htb]
%\centering
%\includegraphics[width=0.49\textwidth]{cifar10_omniscient_small}
%\caption{Top-3 Accuracy VS. \# rounds evaluated on CIFAR10 with Omniscient Attack}
%\label{fig:cifar10_omniscient}
%\end{figure}
%\begin{figure}[htb]
%\centering
%\includegraphics[width=0.49\textwidth]{cifar10_bitflip_small}
%\caption{Top-3 Accuracy VS. \# rounds evaluated on CIFAR10 with Bit-flip Attack}
%\label{fig:cifar10_bitflip}
%\end{figure}
\begin{figure}[htb]
\centering
\includegraphics[width=0.49\textwidth]{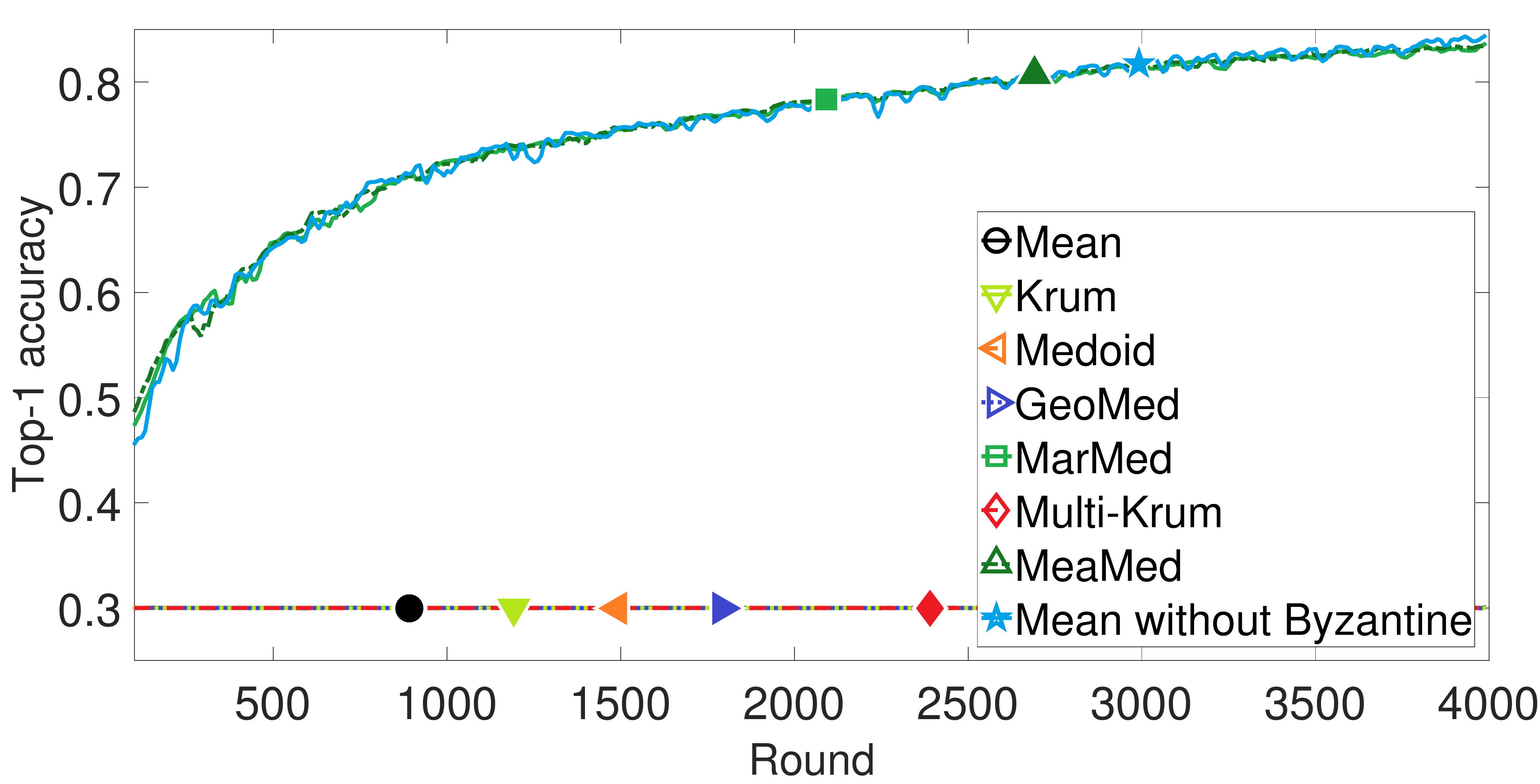}
\caption{Top-3 Accuracy of CNN on CIFAR10 with gambler.}
\label{fig:cifar10_multiserver}
\vspace{-0.4cm}
\end{figure}

%\begin{figure}[htb]
%\centering
%\includegraphics[width=0.49\textwidth]{mlp_bitflip_small}
%\caption{Top-3 Accuracy VS. \# rounds evaluated on MNIST with Bit-flip Attack}
%\label{fig:mlp_bitflip}
%\end{figure}
%\begin{figure}[htb]
%\centering
%\includegraphics[width=0.49\textwidth]{cnn_bitflip_small}
%\caption{Top-3 Accuracy VS. \# rounds evaluated on CIFAR10 with Bit-flip Attack}
%\label{fig:cnn_bitflip}
%\end{figure}

\begin{figure}[htb]
\vspace{-0.3cm}
\centering
\includegraphics[width=0.45\textwidth]{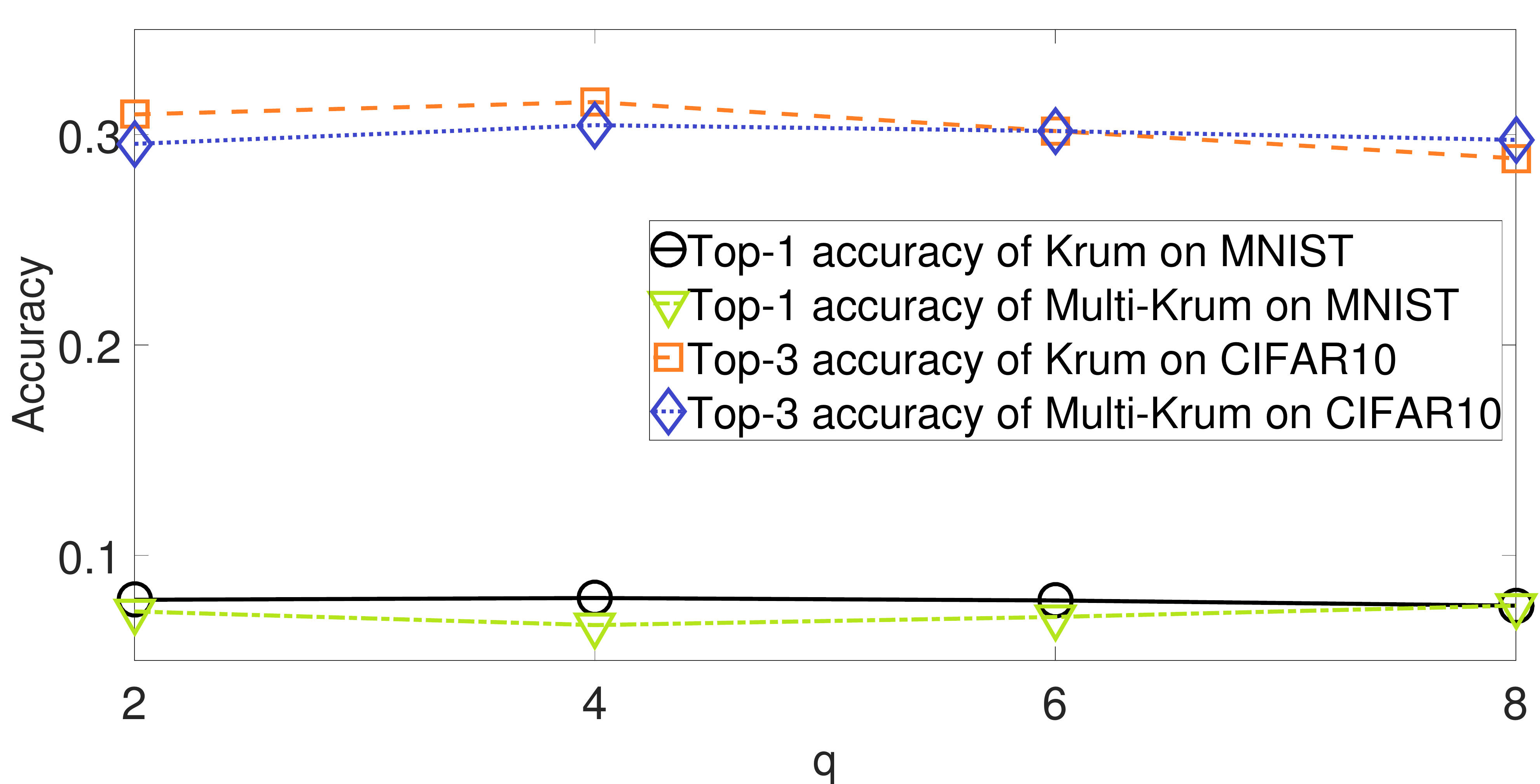}
\caption{Accuracy of $Krum$-based aggregations, at the end of training, when $q$ varies. With 20 servers, $q$ must satisfy $q \leq 8$.}
\label{fig:krum_q}
\vspace{-0.4cm}
\end{figure}

\section{Experiments}
\label{sec:experiments}
In this section, we evaluate the convergence and Byzantine resilience properties of the proposed algorithms. We consider two image classification tasks: handwritten digits classification on MNIST dataset using multi-layer perceptron~(MLP) with two hidden layers, and object recognition on convolutional neural network~(CNN) with five convolutional layers and two fully-connected layers. The details of these two neural networks can be found in the appendix. There are $n=20$ worker processes. We repeat each experiment for ten times and take the average. To make the conditions as fair as possible for all the algorithms, we ensure that all the algorithms are run with the same set of random seeds. The details of the datasets and the default hyperparameters of the corresponding models are listed in Table~\ref{table:datasets}. We use top-1 or top-3 accuracy on testing sets~(disjoint with the training sets) as evaluation metrics.

The baseline aggregation rules are \textit{Mean}, \textit{Medoid}, \textit{Krum}~(Definition~\ref{def:krum}), and \textit{Multi-Krum}. Medoid, defined as follows, is a computation-efficient version of geometric median.
\begin{definition}
\label{def:medoid}
The medoid of $\{\tilde{v}_i: i \in [n]\}$, denoted by $Medoid(\{\tilde{v}_i: i \in [n]\})$, is defined as 
$
Medoid(\{\tilde{v}_i: i \in [n]\}) = \argmin_{v \in \{\tilde{v}_i: i \in [n]\}} \sum_{i=1}^n \|v - \tilde{v}_i\|.
$
\end{definition}
Multi-Krum is a variant of Krum defined in \citet{blanchard2017machine}, which takes the average on several vectors selected by multiple rounds of Krum. 
We compare these baseline algorithms with the proposed algorithms: geometric median~(\textit{GeoMed} defined in Definition~\ref{def:geomed}), marginal median~(\textit{MarMed} defined in Definition~\ref{def:marmed}), and ``mean around median"~(\textit{MeaMed} defined in Definition~\ref{def:meamed}) under different settings in the following subsections.

Note that all the experiments of CNN on CIFAR10 show similar results with the experiments of MLP on MNIST. Thus, we only show the results of CNN in Section~\ref{sec:general_attack} as an example. The remaining results are provided in the appendix.

\subsection{Convergence without Byzantine Failures}
First, we evaluate the convergence without Byzantine failures. The goal is to empirically evaluate the bias and variance caused by the robust aggregation rules. 

In Figure~\ref{fig:mnist_nobyz}, we show the top-1 accuracy on the testing set of MNIST. 
%In Figure~\ref{fig:cifar10_nobyz}, we show the top-3 accuracy on the testing set of CIFAR10. 
The gaps between different algorithms are tiny. Among all the algorithms, \textit{Multi-Krum}, \textit{GeoMed}, and \textit{MeaMed} have the least bias. They act just the same as averaging. $MarMed$ converges slightly slower. \textit{Medoid} and \textit{Krum} both have slowest convergence. 

\subsection{Gaussian Attack}
We test classic Byzantine resilience in this experiment. 
We consider the attackers that replace some of the gradient vectors with Gaussian random vectors with zero mean and isotropic covariance matrix with standard deviation 200. We refer to this kind of attack as \textit{Gaussian Attack}. Within the figure, we also include the averaging without Byzantine failures as a baseline. 6 out of the 20 gradient vectors are Byzantine. The results are shown in Figure~\ref{fig:mnist_gaussian}.
As expected, averaging is not Byzantine resilient. The gaps between all the other algorithms are still tiny. \textit{GeoMed} and \textit{MeaMed} performs like there are no Byzantine failures at all. \textit{Multi-Krum} and \textit{MarMed} converges slightly slower. \textit{Medoid} and \textit{Krum} performs worst. Although \textit{Medoid} is not Byzantine resilient, the Gaussian attack is weak enough so that \textit{Medoid} is still effective.

\subsection{Omniscient Attack}
We test classic Byzantine resilience in this experiment. 
This kind of attacker is assumed to know the all the correct gradients. For each Byzantine gradient vector, the gradient is replaced by the negative sum of all the correct gradients, scaled by a large constant~(1e20 in the experiments). Roughly speaking, this attack tries to make the parameter server go into the opposite direction with a long step. 6 out of the 20 gradient vectors are Byzantine. The results are shown in Figure~\ref{fig:mnist_omniscient}. \textit{MeaMed} still performs just like there is no failure. \textit{Multi-Krum} is not as good as \textit{MeaMed}, but the gap is small.  \textit{Krum} converges slower but still converges to the same accuracy. However, \textit{GeoMed} and \textit{MarBed} converge to bad solutions. \textit{Mean} and \textit{Medoid} are not tolerant to this attack.

\subsection{Bit-flip Attack}
We test dimensional Byzantine resilience in this experiment. 
Knowing the information of other workers can be difficult in practice. Thus, we use more realistic scenario in this experiment. The attacker only manipulates some individual floating numbers by flipping the 22th, 30th, 31th and 32th bits. Furthermore, we test dimensional Byzantine resilience in this experiment. For each of the first 1000 dimensions, 1 of the 20 floating numbers is manipulated using the bit-flip attack. The results are shown in Figure~\ref{fig:mnist_bitflip}. As expected, only \textit{MarMed} and \textit{MeaMed} are dimensional Byzantine resilient.

Note that for \textit{Krum} and \textit{Multi-Krum}, their assumption requires the number of Byzantine vectors $q$ to satisfy $2q + 2 < n$, which means $q \leq 8$ in our experiments. However, because each gradient is partially manipulated, all the $n$ vectors are Byzantine, which breaks the assumption of the Krum-based algorithms. Furthermore, to compute the distances to the $(n-q-2)$-nearest neighbours, $n-q-2$ must be positive. To test the performance of \textit{Krum} and \textit{Multi-Krum}, we set $q = 8$ for these two algorithms so that they can still be executed. 
Furthermore, we test whether tuning $q$ can make a difference. The results are shown in Figure~\ref{fig:krum_q}. Obviously, whatever $q$ we use, \textit{Krum}-based algorithms get stuck around bad solutions. 

\subsection{General Attack with Multiple Servers}
\label{sec:general_attack}
We test general Byzantine resilience in this experiment. 
We evaluate the robust aggregation rules under a more general and realistic type of attack. It is very popular to partition the parameters into disjoint subsets, and use multiple server nodes to storage and aggregate them~\citep{Li2014ScalingDM,Li2014CommunicationED,Ho2013MoreED}. We assume that the parameters are evenly partitioned and assigned to the server nodes. The attacker picks one single server, and manipulates any floating number by multiplying $-1e20$, with probability of $0.05\%$. We call this attack \textit{gambler}, because the attacker randomly manipulate the values, and wish that in some rounds the assumptions/prerequisites of the robust aggregation rules are broken, which crashes the training. Such attack requires less global information, and can be concentrated on one single server, which makes it more realistic and easier to implement.

In Figure~\ref{fig:mnist_multiserver} and \ref{fig:cifar10_multiserver}, we evaluate the performance of all the robust aggregation rules under the gambler attack. The number of servers is $20$. For \textit{Krum}, \textit{Multi-Krum} and \textit{MeaMed}, the estimated Byzantine number $q$ is set as $8$. We also show the performance of averaging without Byzantine values as the benchmark. It is shown that only marginal median \textit{MarMed} and ``mean around median"~\textit{MeaMed} survive under this attack. The convergence is slightly slower than the averaging without Byzantine values, but the gaps are small. 

\subsection{Discussion}
As expected, \textit{mean} aggregation is not Byzantine resilient. Although \textit{medoid} is not Byzantine resilient, as proved by \citet{blanchard2017machine}, it can still make reasonable progress under some attacks such as Gaussian attack. \textit{Krum}, \textit{Multi-Krum}, and \textit{GeoMed} are classic Byzantine resilient but not dimensional Byzantine resilient. \textit{MarMed} and \textit{MeaMed} are dimensional Byzantine resilient. However, under omniscient attack, \textit{MarMed} suffers from larger variances, which slow down the convergence. 

The gambler attack shows the true advantage of dimensional Byzantine resilience: higher probability of survival. Under such attack, chances are that the assumptions/prerequisites of \textit{MarMed} and \textit{MeaMed} may still get broken. However, their probability of crashing is less than the other algorithms because dimensional Byzantine resilience generalizes classic Byzantine resilience. An interesting observation is that \textit{MarMed} is slightly better than \textit{MeaMed} under gambler attack. That is because the estimation of $q = 8$ is not accurate, which will cause some unpredictable behavior for \textit{MeaMed}. We choose $q = 8$ because it is the maximal value we can take for \textit{Krum} and \textit{Multi-Krum}. 
%For \textit{MeaMed}, we can take $q=9$, which makes it behave more similar to \textit{MarMed}.

It is obvious that \textit{MeaMed} performs best in almost all the cases. \textit{Multi-Krum} is also good, except that it is not dimensional Byzantine resilient. The reason why \textit{MeaMed} and \textit{Multi-Krum} have better performance is that they utilize the extra information of the number of Byzantine values. Note that \textit{MeaMed} not only performs just as well as or even better than \textit{Multi-Krum}, but also has lower time complexity. 

Marginal median \textit{MarMed} has the cheapest computation. Its worst case, omniscient attack, is hard to implement in reality. Thus, for most applications, we suggest MarMed as an easy-to-implement aggregation rule with robust performance, which (importantly) does not require knowledge of the number of byzantine values.

\section{Related Works}

There are few papers studying Byzantine resilience for machine learning algorithms. Our work is closely related to \citet{blanchard2017machine}. Another paper~\cite{chen2017distributed} proposed grouped geometric median for Byzantine resilience, with strongly convex functions. 

Our approach offers the following important advantages over the previous work.
\setitemize[0]{leftmargin=*}
\begin{itemize}
\item \textbf{Cheaper computation compared to Krum.} Geometric median has nearly linear~(approximately $O(nd)$) time complexity~\cite{cohen2016geometric}. Marginal median and ``mean around median" have linear time complexity $O(nd)$ on average~\cite{blum1973time}, while the time complexity of Krum is $O(n^2 d)$.
\item \textbf{Less prior knowledge required.} Both geometric median and marginal median do not require $q$, the number of Byzantine workers, to be given, while Krum needs $q$ to calculate the sum of Euclidean distances of the $n-q-2$ nearest neighbours. Furthermore, when $q$ is known or well estimated, \textit{MeaMed} show better robustness than \textit{Krum} and \textit{Multi-Krum} in most cases.
\item \textbf{Dimensional Byzantine resilience.} Marginal median and ``mean around median" tolerate a more general type of Byzantine failures described in Equation~(\ref{equ:byz_model}) and Definition~\ref{def:dim_byz}, while Krum and geometric median can only tolerate the classic Byzantine failures described in Equation~(\ref{equ:byz_worker}) and Definition~\ref{def:byz}.
\item \textbf{Better support for multiple server nodes.} If the entire set of parameters is disjointly partitioned and stored on multiple server nodes, marginal median and ``mean around median" need no additional communication, while Krum and geometric median requires communication among the server nodes.
\end{itemize}

\section{Conclusion}
We investigate the generalized Byzantine resilience of parameter server architecture. We proposed three novel median-based aggregation rules for synchronous SGD. The algorithms have low time complexity and provable convergence to critical points. Our empirical results show good performance in practice. 

\newpage
\bibliography{byz}
\bibliographystyle{icml2018}

\newpage
\clearpage
\section{Appendix}
In the appendix, we introduce several useful lemmas and use them to derive the detailed proofs of the theorems in this paper.

\subsection{Dimensional Byzantine Resilience}
\setcounter{theorem}{0}

\begin{theorem}
\label{thm:mean_dim_byz}
Averaging is not dimensional Byzantine resilient.
\end{theorem}
\begin{proof}
We demonstrate a counter example.
Consider the case where
\begin{align}
\tilde{v}_i = 
\begin{cases}
v_i, &\forall i \in [n-1]\\
-g - \sum_{i=1}^{n-1} v_i, &i=n,
\end{cases}
\end{align}
where $g = \E[v_i]$, $\forall i \in [n]$. Thus, the resulting aggregation is $Aggr = -g/n$.
The inner product $\ip{\E[Aggr]}{g}$ is always negative under the Byzantine attack. Thus, SGD is not expectedly descendant, which means it will not converge to critical points. Note that in this counter example, the number of Byzantine values of each dimension is $1$.

Hence,  averaging is not dimensional $(\alpha, q)$-Byzantine resilient with $\forall \alpha, \forall q > 0$.
\end{proof}

\begin{theorem}
Any aggregation rule $Aggr(\{\tilde{v}_i: i \in [n]\})$ that outputs $Aggr \in \{\tilde{v}_i: i \in [n]\}$ is not dimensional Byzantine resilient.
\end{theorem}
\begin{proof}
We demonstrate a counter example.
Consider the case where the $i$th dimension of the $i$th vector $v_i$ is manipulated by the malicious workers (e.g. multiplied by an arbitrarily large negative value), where $i \in [n]$. Thus, up to 1 value of each dimension is Byzantine. However, no matter which vector is chosen, as long as the aggregation is chosen from $\{\tilde{v}_i: i \in [n]\}$, the inner product $\ip{\E[Aggr]}{g}$ can be arbitrarily large negative value under the Byzantine attack. Thus, SGD is not expectedly descendant, which means it will not converge to critical points.

Hence,  any aggregation rule that outputs $Aggr \in \{\tilde{v}_i: i \in [n]\}$ is not dimensional $(\alpha, q)$-Byzantine resilient with $\forall \alpha, \forall q > 0$.
\end{proof}

\subsection{Geometric Median}
We use the following lemma~\cite{minsker2015geometric,cohen2016geometric} without proof to bound the geometric median.
\begin{lemma}
\label{lem:geo_median}
Let $z_1, \ldots, z_n$ denote $n$ points in a Hilbert space. Let $z_*$ denote a $(1+\epsilon)$-approximation of their geometric median, i.e., $\sum_{i \in [n]} \|z_*-z_i\| \leq (1+\epsilon) \min_z \sum_{i \in [n]} \|z-z_i\|$ for $\epsilon \geq 0$. For any $q$ such that $\frac{q}{n} \in (0, 1/2)$ and given $r \in \R$, if $\sum_{i \in [n]} \1_{\|z_i\| \leq r} \geq (1-q/n)n$, then 
\begin{align*}
\|z_*\| \leq c_q r + \epsilon c_z, 
\end{align*}
where $c_q = \frac{2n-2q}{n-2q}$, $c_z = \frac{\min_z \sum_{i \in [n]} \|z-z_i\|}{n-2q}$.
\end{lemma}
Ideally, the geometric median~($\epsilon = 0$) ignores the second term $\epsilon c_z$.

Using the lemma above, we can prove the classic Byzantine resilience of geometric median.
\begin{theorem}
Let $v_1, \ldots, v_n$ be any i.i.d. random $d$-dimensional vectors s.t. $v_i \sim G$, with $\E[G] = g$ and $\E\|G-g\|^2 = d\sigma^2$. $q$ of $\{v_i: i \in [n]\}$ are replaced by arbitrary $d$-dimensional vectors $b_1, \ldots, b_q$. If $q \leq \ceil{\frac{n}{2}}-1$ and $\eta_1(n, q)\sqrt{d}\sigma < \|g\|$, where $\eta_1(n, q) = \frac{2n-2q}{n-2q} \sqrt{n-q}$, then the $GeoMed$ function is classic $(\alpha_1, q)$-Byzantine resilient where $0 \leq \alpha_1 < \pi/2$ is defined by 
$\sin \alpha_1 = \frac{\eta_1(n,q) \sqrt{d} \sigma}{\|g\|}$.
\end{theorem}
\begin{proof}
We only need to prove that $GeoMed(\cdot)$ satisfies the two conditions of classic $(\alpha_1, q)$-Byzantine resilience defined in Definition~\ref{def:byz}.

\textbf{Condition (i)}:\\
Let the sequence $\{\tilde{v}_j: j \in [n]\}$ be defined as
\begin{align*}
\tilde{v}_j = 
\begin{cases}
v_j, \mbox{for correct $j$,}\\
arbitrary, \mbox{for Byzantine $j$}.
\end{cases}
\end{align*}
Let $\lambda$ denote the geometric median of $\{\tilde{v}_j: j \in [n]\}$. Thus, $z_* = \lambda - g$ is the geometric median of $\{\tilde{v}_j - g: j \in [n]\}$. Using Lemma~\ref{lem:geo_median}, and taking $r = \max_{\mbox{correct } j}\|\tilde{v}_j - g\|$, under the assumption 
$q \leq \ceil{\frac{n}{2}}-1 < n/2$, we obtain 
\begin{align*}
\|\lambda - g\| \leq \frac{2n-2q}{n-2q} \max_{\mbox{correct } j}\|\tilde{v}_j - g\|.
\end{align*}

Now, we can bound $\|\E[\lambda] - g\|^2$ as follows:
\begin{align*}
&\|\E[\lambda] - g\|^2 \\
&\leq \E \|\lambda - g\|^2 \quad \mbox{(Jensen's inequality)} \\
&\leq \E \left[ \left(\frac{2n-2q}{n-2q}\right)^2 \max_{\mbox{correct } j}\|\tilde{v}_j - g\|^2 \right] \\
&\leq \E \left[ \left(\frac{2n-2q}{n-2q}\right)^2 \sum_{\mbox{correct } j}\|\tilde{v}_j - g\|^2 \right] \\
&= \underbrace{\left(\frac{2n-2q}{n-2q}\right)^2 (n-q)}_{\eta_1^2(n, q)} d\sigma^2.
\end{align*}
By assumption, $\eta_1(n,q) \sqrt{d} \sigma < \|g\|$, i.e. $\E[\lambda]$ belongs to a ball centered at $g$ with radius $\eta_1(n,q) \sqrt{d} \sigma$. This implies 
\begin{align*}
\ip{\E[\lambda]}{g} \geq (1- \sin^2 \alpha_1) \|g\|^2 \geq (1- \sin \alpha_1) \|g\|^2,
\end{align*}
where $\sin \alpha_1 = \eta_1(n,q) \sqrt{d} \sigma / \|g\|$.

\textbf{Condition (ii)}:\\
We re-use Lemma~\ref{lem:geo_median} by taking $z_* = \lambda$, $z_i = \tilde{v}_i$ for $\forall i \in [n]$, and $r = \max_{\mbox{correct } j}\|\tilde{v}_j\|$. Thus, we have
\begin{align*}
&\|\lambda\| 
&\leq c_q \max_{\mbox{correct } j}\|\tilde{v}_j\| 
&\leq c_q \sum_{\mbox{correct } j} \|\tilde{v}_j\|. 
\end{align*}
Without loss of generality, we denote the sequence $\{\tilde{v}_j: \mbox{correct } j\}$ as $\{v_1, \ldots, v_{n-q}\}$. 
Thus, there exists a constant $c_0$ such that
\begin{align*}
\|\lambda\|^r \leq c_0 \sum_{r_1+\ldots+r_{n-q} = r} \|v_1\|^{r_1} \ldots \|v_{n-q}\|^{r_{n-q}}.
\end{align*}
Since $v_i$'s are i.i.d., we obtain that $\E\|\lambda\|^r$ is bounded above by a linear combination of terms of the form $\E\|v_1\|^{r_1} \ldots \E\|v_{n-q}\|^{r_{n-q}} = \E\|G\|^{r_1} \ldots \E\|G\|^{r_{n-q}}$ with $r_1+\ldots+r_{n-q} = r$, which completes the proof of condition (ii).
\end{proof}

\subsection{Marginal Median}
We use the following lemma to bound the one-dimensional median.
\begin{lemma}
\label{lem:median}
For a sequence composed of $q$ Byzantine values and $n-q$ correct values $u_1, \ldots, u_{n-q}$, if $q \leq \ceil{\frac{n}{2}} - 1$~(the correct value dominates the sequence), then the median value $m$ of this sequence satisfies $m \in [\min_i u_i, \max_i u_i]$, $i \in [n]$.
\end{lemma}
\begin{proof}
If $m$ comes from correct values, then the result is trivial. Thus, we only need to consider the cases where $m$ comes from Byzantine values.

If $n$ is odd, then in the sorted sequence, there will be $\frac{n-1}{2}$ values on both sides of $m$. However, the number of correct values $n-q \geq \frac{n+1}{2} > \frac{n-1}{2}$. Thus, on both sides of $m$, there will be at least one correct value, which yields the desired result.

Furthermore, if $n$ is even, we can re-use the same technique above to prove $m \in [\min_i u_i, \max_i u_i]$.
\end{proof}

\begin{theorem}
Let $v_1, \ldots, v_n$ be any i.i.d. random $d$-dimensional vectors s.t. $v_i \sim G$, with $\E[G] = g$ and $\E\|G-g\|^2 = d\sigma^2$. For any dimension $j \in [d]$, $q$ of $\{(v_1)_j, \dots, (v_n)_j\}$ are replaced by arbitrary values, where $(v_i)_j$ is the $j$th dimension of the vector $v_i$. If $q \leq \ceil{\frac{n}{2}}-1$ and $\eta_2(n, q)\sqrt{d}\sigma < \|g\|$, where $\eta_2(n, q) = \sqrt{n-q}$, then the $MarMed$ function is dimensional $(\alpha_2, q)$-Byzantine resilient where $0 \leq \alpha_2 < \pi/2$ is defined by 
$\sin \alpha_2 = \frac{\eta_2(n,q) \sqrt{d} \sigma}{\|g\|}$.
\end{theorem}
\begin{proof}
We only need to prove that $MarMed(\cdot)$ satisfies the two conditions of dimensional $(\alpha_2, q)$-Byzantine resilience defined in Definition~\ref{def:dim_byz}.

\textbf{Condition (i)}:\\
Without loss of generality, we assume that $\E[G_i - g_i]^2 = \sigma_i^2$, $\E\|G-g\|^2 = \E\sum_{i=1}^d [G_i-g_i]^2 = \sum_{i=1}^d \sigma_i^2 = d\sigma^2$.
For any dimension $j \in [d]$, let the sequence $\{(\tilde{v}_1)_j, \ldots, (\tilde{v}_n)_j\}$ be defined as
\begin{align*}
(\tilde{v}_i)_j = 
\begin{cases}
(v_i)_j, \mbox{for correct $j$,}\\
arbitrary, \mbox{for Byzantine $j$}.
\end{cases}
\end{align*}
For the $j$th dimension, $j \in [d]$, the median value $\mu_j \in [\min_{\mbox{correct } i} (\tilde{v}_i)_j, \max_{\mbox{correct } i} (\tilde{v}_i)_j]$. 

Thus, we have 
\begin{align*}
&\E[\mu_j - g_j]^2 
\leq \E\left[ \max_{\mbox{correct } i}((\tilde{v}_i)_j - g_j)^2 \right] \\
&\leq \E\left[ \sum_{\mbox{correct } i}((\tilde{v}_i)_j - g_j)^2 \right] 
= \sum_{\mbox{correct } i} \E\left[ ((\tilde{v}_i)_j - g_j)^2 \right] \\
&= (n-q) \E[G_j - g_j]^2 \quad \mbox{(i.i.d. over $i$)} \\
&= (n-q) \sigma_j^2.
\end{align*}
Now, we can bound $\|\E[\mu] - g\|^2$ as follows:
\begin{align*}
&\|\E[\mu] - g\|^2 
\leq \E \|\mu - g\|^2 \quad \mbox{(Jensen's inequality)} \\
&= \E \left[ \sum_{j=1}^d (\mu_j-g_j)^2 \right] 
= \sum_{j=1}^d \E \left[ (\mu_j-g_j)^2 \right] \\
&\leq \sum_{j=1}^d (n-q) \sigma_j^2 
= (n-q) \sum_{j=1}^d \sigma_j^2 
= \underbrace{(n-q)}_{\eta_2^2(n, q)} d\sigma^2.
\end{align*}
By assumption, $\eta_2(n,q) \sqrt{d} \sigma < \|g\|$, i.e. $\E[\mu]$ belongs to a ball centered at $g$ with radius $\eta_2(n,q) \sqrt{d} \sigma$. This implies 
\begin{align*}
\ip{\E[\mu]}{g} \geq (1- \sin^2 \alpha_2) \|g\|^2 \geq (1- \sin \alpha_2) \|g\|^2,
\end{align*}
where $\sin \alpha_2 = \eta_2(n,q) \sqrt{d} \sigma / \|g\|$.

\textbf{Condition (ii)}:\\
By using the equivalence of norms in finite dimension, there exists a constant $c_1$ such that
\begin{align*}
&\|\mu\| = \sqrt{\sum_{j=1}^d \mu_j^2} 
\leq \sqrt{\sum_{j=1}^d \max_{\mbox{correct } i} (\tilde{v}_i)_j^2} \\
&\leq \sqrt{\sum_{j=1}^d \sum_{\mbox{correct } i} (\tilde{v}_i)_j^2} 
= \sqrt{\sum_{\mbox{correct } i} \|\tilde{v}_i\|^2} \\
&\leq c_1 \sum_{\mbox{correct } i} \|\tilde{v}_i\|. \\
&\mbox{(equivalence between $\ell_2$-norm and $\ell_1$-norm)}
\end{align*}
Without loss of generality, we denote the sequence $\{\tilde{v}_i: \mbox{correct } i\}$ as $\{v_1, \ldots, v_{n-q}\}$. 
Thus, there exists a constant $c_2$ such that
\begin{align*}
\|\mu\|^r \leq c_2 \sum_{r_1+\ldots+r_{n-q} = r} \|v_1\|^{r_1} \ldots \|v_{n-q}\|^{r_{n-q}}.
\end{align*}
Since $v_i$'s are i.i.d., we obtain that $\E\|\mu\|^r$ is bounded above by a linear combination of terms of the form $\E\|v_1\|^{r_1} \ldots \E\|v_{n-q}\|^{r_{n-q}} = \E\|G\|^{r_1} \ldots \E\|G\|^{r_{n-q}}$ with $r_1+\ldots+r_{n-q} = r$, which completes the proof of condition (ii).
\end{proof}

\subsection{Mean around Median}

The following lemma bounds the one-dimensional mean around median.
\begin{lemma}
\label{lem:meamed}
For a sequence (of scalar values) composed of $q$ Byzantine values and $n-q$ correct values $u_1, \ldots, u_{n-q}$, if $q \leq \ceil{\frac{n}{2}} - 1$~(the correct value dominates the sequence), then the mean-around-median value $\rho$~(defined in Definition~\ref{def:meamed}) and the median $\mu$~(defined in Definition~\ref{def:marmed}) of this sequence satisfies $|\rho - \mu| \leq \max_{i} |u_i - \mu|$.
\end{lemma}
\begin{proof}
According to the definition of the mean around median $\rho$, it is the mean value over the top-$(n-1)$ values in the sequence, nearest to the median $\mu$. Denote such set of nearest values as $\{w_1, \ldots, w_{n-q}\}$. If any $w_i$ satisfies that $|w_i - \mu| > \max_{i} |u_i - \mu|$, then it cannot be in the set of the top-$(n-q)$ nearest values because all the $n-q$ correct values are nearer to $\mu$~($|u_i - \mu| \leq \max_{i} |u_i - \mu|$). Since all $w_i$ satisfies $|w_i - \mu| \leq \max_{i} |u_i - \mu|$, the average over them must also satisfies $|\frac{1}{n-q}\sum_i w_i - \mu| \leq \max_{i} |u_i - \mu|$.
\end{proof}

\begin{theorem}
Let $v_1, \ldots, v_n$ be any i.i.d. random $d$-dimensional vectors s.t. $v_i \sim G$, with $\E[G] = g$ and $\E\|G-g\|^2 = d\sigma^2$. For any dimension $j \in [d]$, $q$ of $\{(v_1)_j, \dots, (v_n)_j\}$ are replaced by arbitrary values, where $(v_i)_j$ is the $j$th dimension of the vector $v_i$. If $q \leq \ceil{\frac{n}{2}}-1$ and $\eta_3(n, q)\sqrt{d}\sigma < \|g\|$, where $\eta_3(n, q) = \sqrt{10(n-q)}$, then the $MeaMed$ function is dimensional $(\alpha_3, q)$-Byzantine resilient where $0 \leq \alpha_3 < \pi/2$ is defined by 
$\sin \alpha_3 = \frac{\eta_3(n,q) \sqrt{d} \sigma}{\|g\|}$.
\end{theorem}
\begin{proof}
We only need to prove that $MeaMed(\cdot)$ satisfies the two conditions of $(\alpha_3, q)$-Byzantine resilience defined in Definition~\ref{def:dim_byz}.

\textbf{Condition (i)}:\\
Without loss of generality, we assume that $\E[G_i - g_i]^2 = \sigma_i^2$, $\E\|G-g\|^2 = \E\sum_{i=1}^d [G_i-g_i]^2 = \sum_{i=1}^d \sigma_i^2 = d\sigma^2$.
For any dimension $j \in [d]$, let the sequence $\{(\tilde{v}_1)_j, \ldots, (\tilde{v}_n)_j\}$ be defined as
\begin{align*}
(\tilde{v}_i)_j = 
\begin{cases}
(v_i)_j, \mbox{for correct $j$,}\\
arbitrary, \mbox{for Byzantine $j$}.
\end{cases}
\end{align*}
For the $j$th dimension, $j \in [d]$, using Lemma~\ref{lem:meamed}, we have $|\rho_j - \mu_j| \leq \max_{\mbox{correct } i} |(\tilde{v}_i)_j - \mu_j|$, where $\mu_j$ is the median of the $j$th dimension.

Thus, we have 
\begin{align*}
&\E[\rho_j - g_j]^2 \\
&\leq 2\E[\rho_j - \mu_j]^2 + 2\E[\mu_j - g_j]^2 \\
&\leq 2\E\max_{\mbox{correct } i} [(\tilde{v}_i)_j - \mu_j]^2 + 2\E[\mu_j - g_j]^2 \\
&\leq 4\E\max_{\mbox{correct } i} [(\tilde{v}_i)_j - g_j]^2 + 6\E[\mu_j - g_j]^2 \\
&\leq 10\E\left[ \max_{\mbox{correct } i}((\tilde{v}_i)_j - g_j)^2 \right] \\
&\leq 10\E\left[ \sum_{\mbox{correct } i}((\tilde{v}_i)_j - g_j)^2 \right] \\
&= 10\sum_{\mbox{correct } i} \E\left[ ((\tilde{v}_i)_j - g_j)^2 \right] \\
&= 10(n-q) \E[G_j - g_j]^2 \quad \mbox{(i.i.d. over $i$)} \\
&= 10(n-q) \sigma_j^2.
\end{align*}
Now, we can bound $\|\E[\rho] - g\|^2$ as follows:
\begin{align*}
&\|\E[\rho] - g\|^2 
\leq \E \|\rho - g\|^2 \quad \mbox{(Jensen's inequality)} \\
&= \E \left[ \sum_{j=1}^d (\rho_j-g_j)^2 \right] 
= \sum_{j=1}^d \E \left[ (\rho_j-g_j)^2 \right] \\
&\leq \sum_{j=1}^d 10(n-q) \sigma_j^2 
= 10(n-q) \sum_{j=1}^d \sigma_j^2 
= \underbrace{10(n-q)}_{\eta_3^2(n, q)} d\sigma^2.
\end{align*}
By assumption, $\eta_3(n,q) \sqrt{d} \sigma < \|g\|$, i.e. $\E[\rho]$ belongs to a ball centered at $g$ with radius $\eta_3(n,q) \sqrt{d} \sigma$. This implies 
\begin{align*}
\ip{\E[\rho]}{g} \geq (1- \sin^2 \alpha_3) \|g\|^2 \geq (1- \sin \alpha_3) \|g\|^2,
\end{align*}
where $\sin \alpha_3 = \eta_3(n,q) \sqrt{d} \sigma / \|g\|$.

\textbf{Condition (ii)}:\\
By using the equivalence of norms in finite dimension, there exists a constant $c_3$ such that
\begin{align*}
&\|\rho\| = \sqrt{\sum_{j=1}^d \rho_j^2} \\
&\leq \sqrt{\sum_{j=1}^d 2[\rho_j - \mu_j]^2 + 2\mu_j^2} \\
&\leq \sqrt{\sum_{j=1}^d \max_{\mbox{correct } i} 2[(\tilde{v}_i)_j - \mu_j]^2 + 2\mu_j^2} \\
&\leq \sqrt{\sum_{j=1}^d 10\max_{\mbox{correct } i} (\tilde{v}_i)_j^2} \\
&\leq \sqrt{10\sum_{j=1}^d \sum_{\mbox{correct } i} (\tilde{v}_i)_j^2} 
= \sqrt{10\sum_{\mbox{correct } i} \|\tilde{v}_i\|^2} \\
&\leq c_3 \sum_{\mbox{correct } i} \|\tilde{v}_i\|. \\
&\mbox{(equivalence between $\ell_2$-norm and $\ell_1$-norm)}
\end{align*}
Without loss of generality, we denote the sequence $\{\tilde{v}_i: \mbox{correct } i\}$ as $\{v_1, \ldots, v_{n-q}\}$. 
Thus, there exists a constant $c_4$ such that
\begin{align*}
\|\rho\|^r \leq c_4 \sum_{r_1+\ldots+r_{n-q} = r} \|v_1\|^{r_1} \ldots \|v_{n-q}\|^{r_{n-q}}.
\end{align*}
Since $v_i$'s are i.i.d., we obtain that $\E\|\rho\|^r$ is bounded above by a linear combination of terms of the form $\E\|v_1\|^{r_1} \ldots \E\|v_{n-q}\|^{r_{n-q}} = \E\|G\|^{r_1} \ldots \E\|G\|^{r_{n-q}}$ with $r_1+\ldots+r_{n-q} = r$, which completes the proof of condition (ii).
\end{proof}

\newpage
\clearpage
\subsection{Experimental Details}
\label{sec:sup_exp}
In Table~\ref{tbl:mlp} and \ref{tbl:cnn}, we show the detailed network structures of the MLP and CNN used in our experiments.
\begin{table}[htb!]
\label{tbl:mlp}
\caption{MLP Summary}
\begin{tabular}{|l|l|l|}
\hline 
Layer (type) & Parameters & Previous Layer \\ \hline 
flatten(Flatten) & null & data \\ \hline
fc1(FullyConnected) & \#output=128 & flatten \\ \hline
relu1(Activation) & null & fc1 \\ \hline
fc2(FullyConnected) & \#output=128 & relu1 \\ \hline
relu2(Activation) & null & fc2 \\ \hline
fc3(FullyConnected) & \#output=10 & relu2 \\ \hline
softmax(SoftmaxOutput) & null & fc3 \\ \hline
\end{tabular} 
\end{table}
\begin{table}[htb!]
\label{tbl:cnn}
\centering
\caption{CNN Summary}
\begin{tabular}{|l|l|l|}
\hline 
Layer (type) & Parameters & Previous Layer \\ \hline 
conv1(Convolution)& channels=32, kernel\_size=3, padding=1 &data \\ \hline 
activation1(Activation)& null &conv1 \\ \hline 
conv2(Convolution)& channels=32, kernel\_size=3, padding=1 &activation1 \\ \hline 
activation2(Activation)& null &conv2 \\ \hline 
pooling1(Pooling)& pool\_size=2 &activation2 \\ \hline 
dropout1(Dropout)& probability=0.2 &pooling1 \\ \hline 
conv3(Convolution)& channels=64, kernel\_size=3, padding=1 &dropout1 \\ \hline 
activation2(Activation)& null &conv3 \\ \hline 
conv4(Convolution)& channels=64, kernel\_size=3, padding=1 &activation2 \\ \hline 
activation4(Activation)& null &conv4 \\ \hline 
pooling2(Pooling)& pool\_size=2 &activation4 \\ \hline 
dropout2(Dropout)& probability=0.2 &pooling2 \\ \hline 
flatten1(Flatten)& null &dropout2 \\ \hline 
fc1(FullyConnected)& \#output=512 &flatten1 \\ \hline 
activation5(Activation)& null &fc1 \\ \hline 
dropout3(Dropout)& probability=0.2 &activation5 \\ \hline 
fc2(FullyConnected)& \#output=512 &dropout3 \\ \hline 
activation6(Activation)& null &fc2 \\ \hline 
dropout4(Dropout)& probability=0.2 &activation6 \\ \hline 
fc3(FullyConnected)& \#output=10 &dropout4 \\ \hline 
softmax(SoftmaxOutput)& null &fc3 \\ \hline   
\end{tabular} 
\end{table}

\newpage
\clearpage
\subsection{Additional Experiments}

In this section, we illustrate the additional empirical results. 

In Figure~\ref{fig:mnist_batchsize}, we illustrate the top-1 accuracy of MLP on MNIST when batch-size varies, without Byzantine failures. The learning rate is 
\begin{align*}
\gamma = \frac{0.1 \times batchsize}{32}.
\end{align*}
The results show that when there is no Byzantine failures, \textit{GeoMed}, \textit{Multi-Krum}, and \textit{MeaMed} performs just like \textit{Mean}. \textit{MarMed} has slightly slower convergence. \textit{Krum} and \textit{Medoid} are the slowest. The gap is narrowed when the batch size increases.

In Figure~\ref{fig:mnist_nobyz}, we illustrate the top-1 accuracy of MLP on MNIST with gambler attack, when the estimated $q$ varies for \textit{Krum}, \textit{Multi-Krum}, and \textit{MeaMed}. \textit{Mean} without Byzantine failures and \textit{MarMed} are used as baselines. No matter what $q$ we use, the Krum-based algorithms always crash. For \textit{MeaMed}, when the estimated $q$ is too small (e.g., $q=2$), it will also crash. In most cases, \textit{MeaMed} performs well. The performance of \textit{MeaMed} is similar to \textit{MarMed}.

We illustrate all the experimental results of CNN on CIFAR10 additional to Section~\ref{sec:experiments}. For completeness, we also illustrate the experimental results of MLP on MNIST. The results are shown in Figure~\ref{fig:mnist_nobyz_appendix}-\ref{fig:cifar10_multiserver_appendix}. In general, all the experiments of CNN on CIFAR10 show similar results with the experiments of MLP on MNIST.

\begin{figure*}[htb!]
\centering
\subfigure[MLP on MNIST without Byzantine with different batch sizes]{\includegraphics[width=0.90\textwidth]{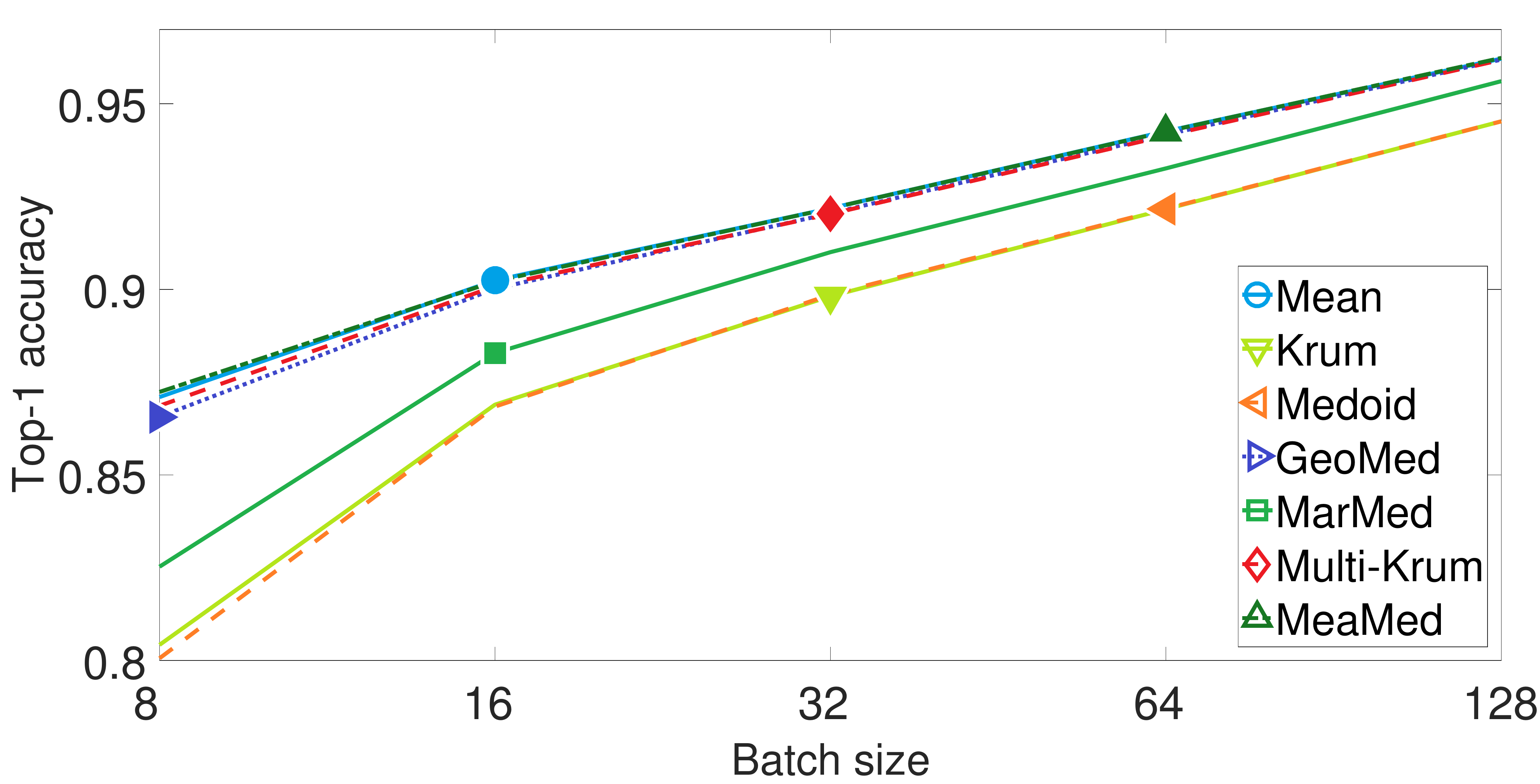}}
\caption{Top-1 accuracy of MLP on MNIST without Byzantine failures, when batch size varies. The learning rate is $\gamma = \frac{0.1 \times batchsize}{32}$.}
\label{fig:mnist_batchsize}
\end{figure*}

\begin{figure*}[htb!]
\centering
\subfigure[MLP on MNIST with gambler]{\includegraphics[width=0.90\textwidth]{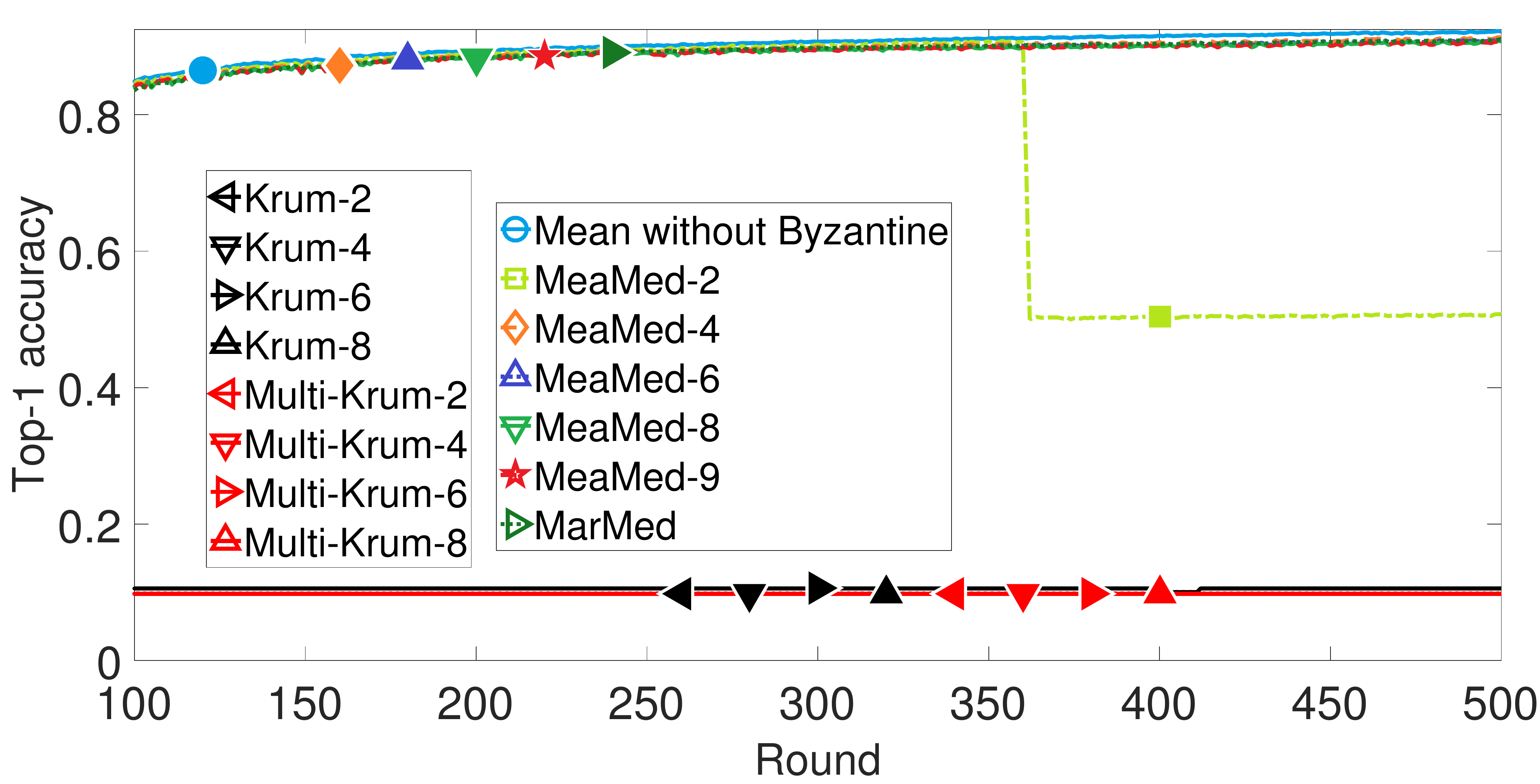}}
\subfigure[MLP on MNIST with gambler~(zoomed)]{\includegraphics[width=0.90\textwidth]{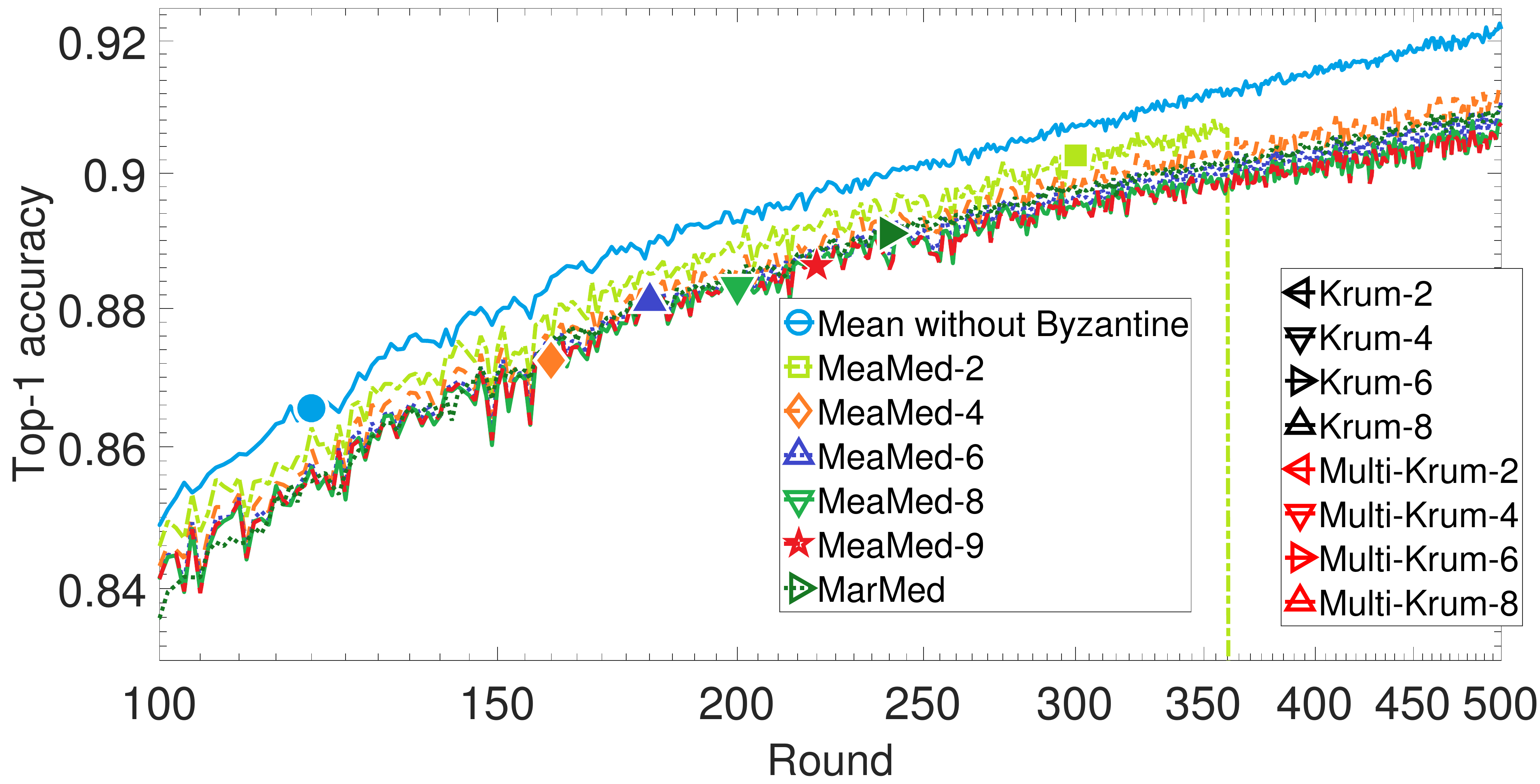}}
\caption{Top-1 accuracy of MLP on MNIST with gambler attack, when $q$ varies}
\label{fig:mnist_nobyz}
\end{figure*}

\begin{figure*}[htb!]
\centering
\subfigure[MLP on MNIST without Byzantine]{\includegraphics[width=0.90\textwidth]{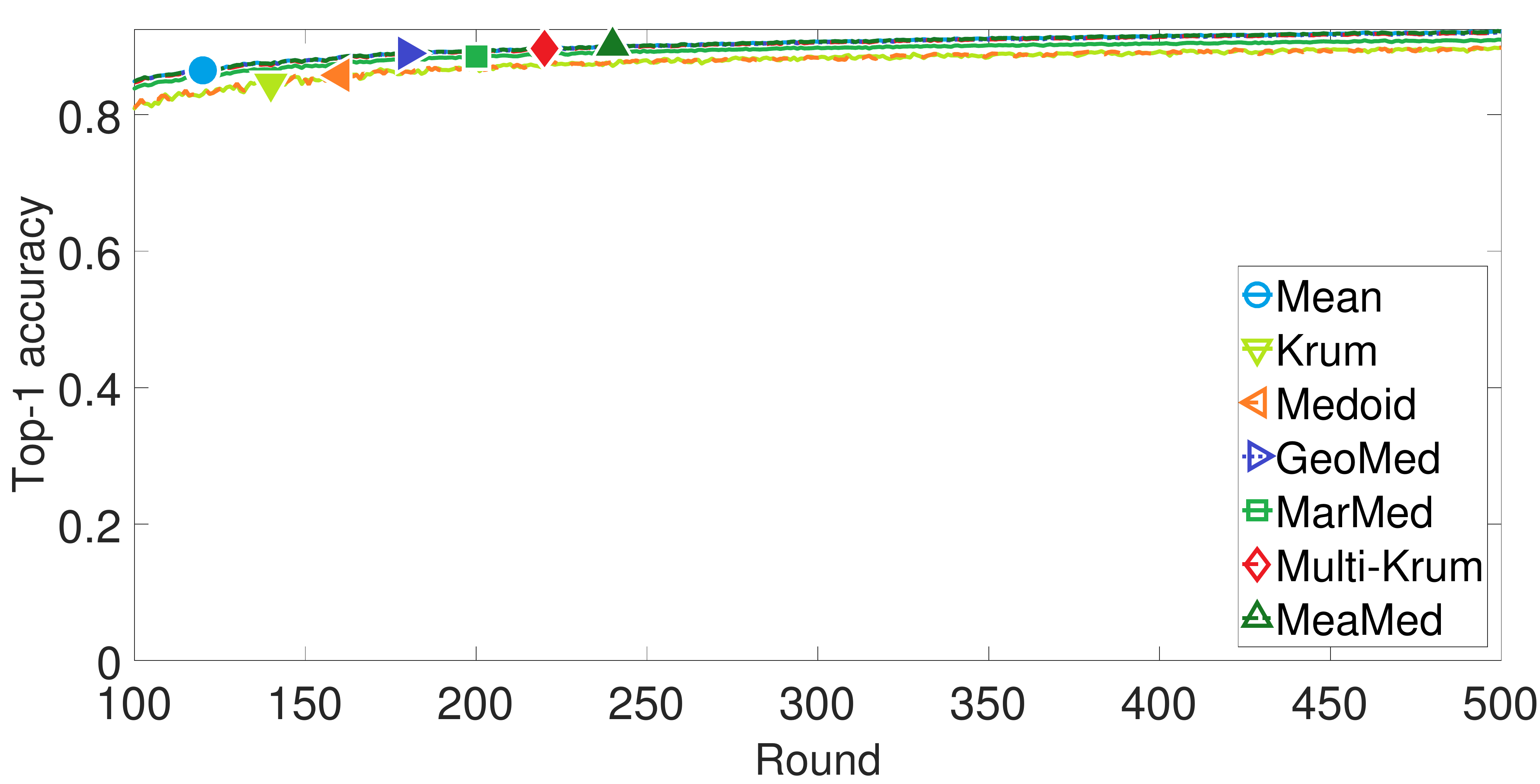}}
\subfigure[MLP on MNIST without Byzantine~(zoomed)]{\includegraphics[width=0.90\textwidth]{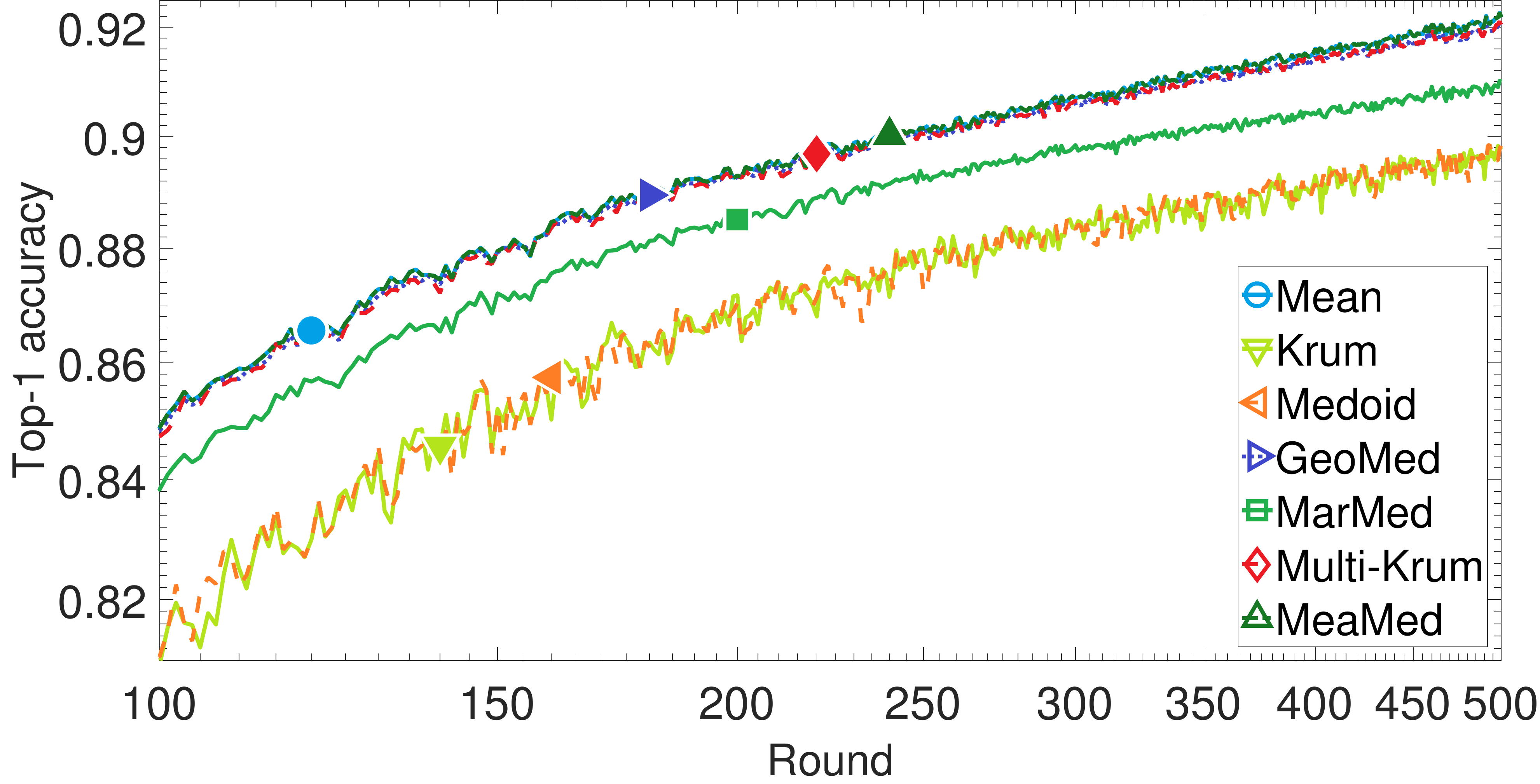}}
\caption{Top-1 accuracy of MLP on MNIST without Byzantine failures.}
\label{fig:mnist_nobyz_appendix}
\end{figure*}
\begin{figure*}[htb!]
\centering
\subfigure[MLP on MNIST with Gaussian]{\includegraphics[width=0.90\textwidth]{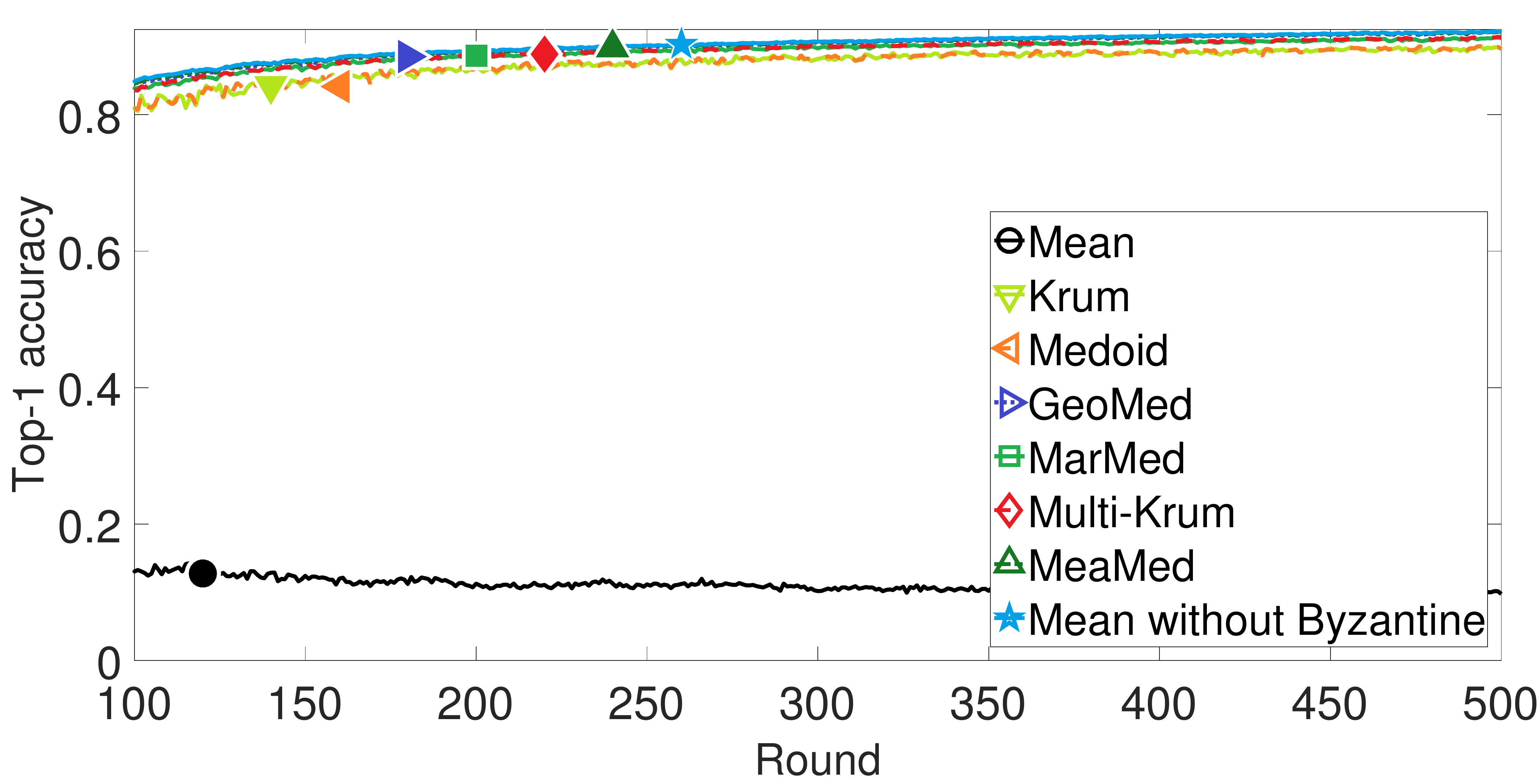}}
\subfigure[MLP on MNIST with Gaussian~(zoomed)]{\includegraphics[width=0.90\textwidth]{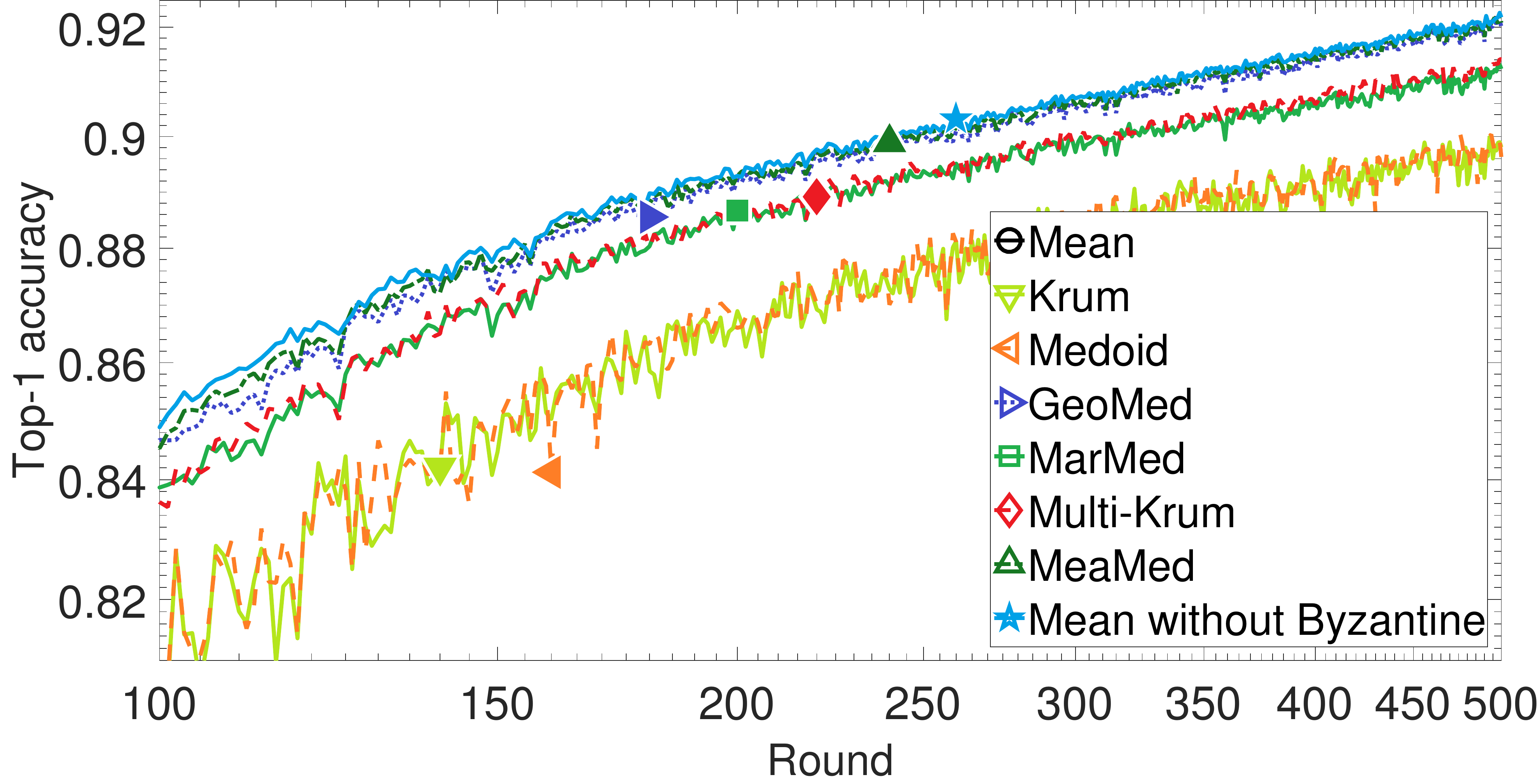}}
\caption{Top-1 accuracy of MLP on MNIST with Gaussian Attack. 6 out of 20 gradient vectors are replaced by i.i.d. random vectors drawn from a Gaussian distribution with 0 mean and 200 standard deviation.}
\label{fig:mnist_gaussian_appendix}
\end{figure*}
\begin{figure*}[htb!]
\centering
\subfigure[MLP on MNIST with omniscient]{\includegraphics[width=0.90\textwidth]{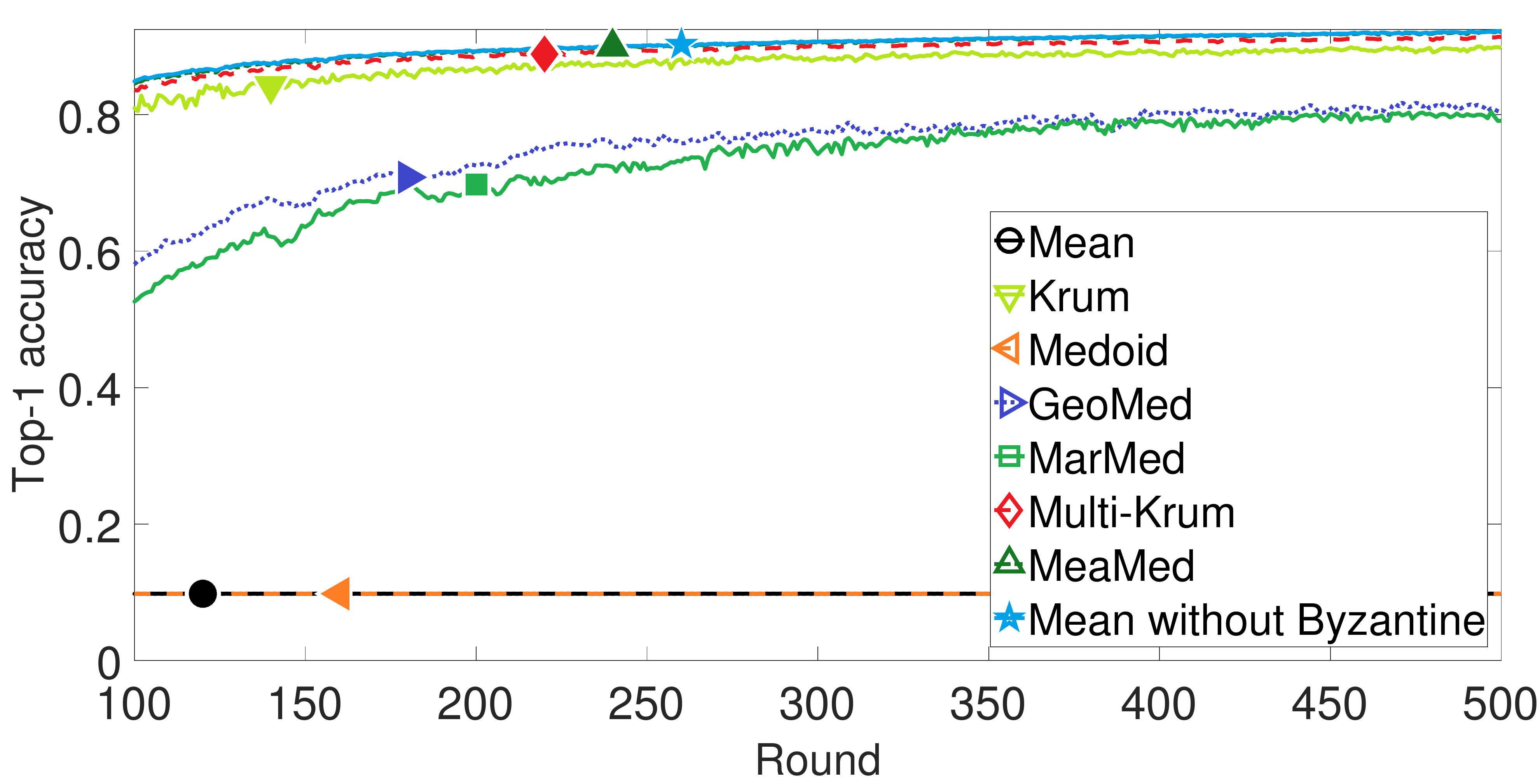}}
\subfigure[MLP on MNIST with omniscient~(zoomed)]{\includegraphics[width=0.90\textwidth]{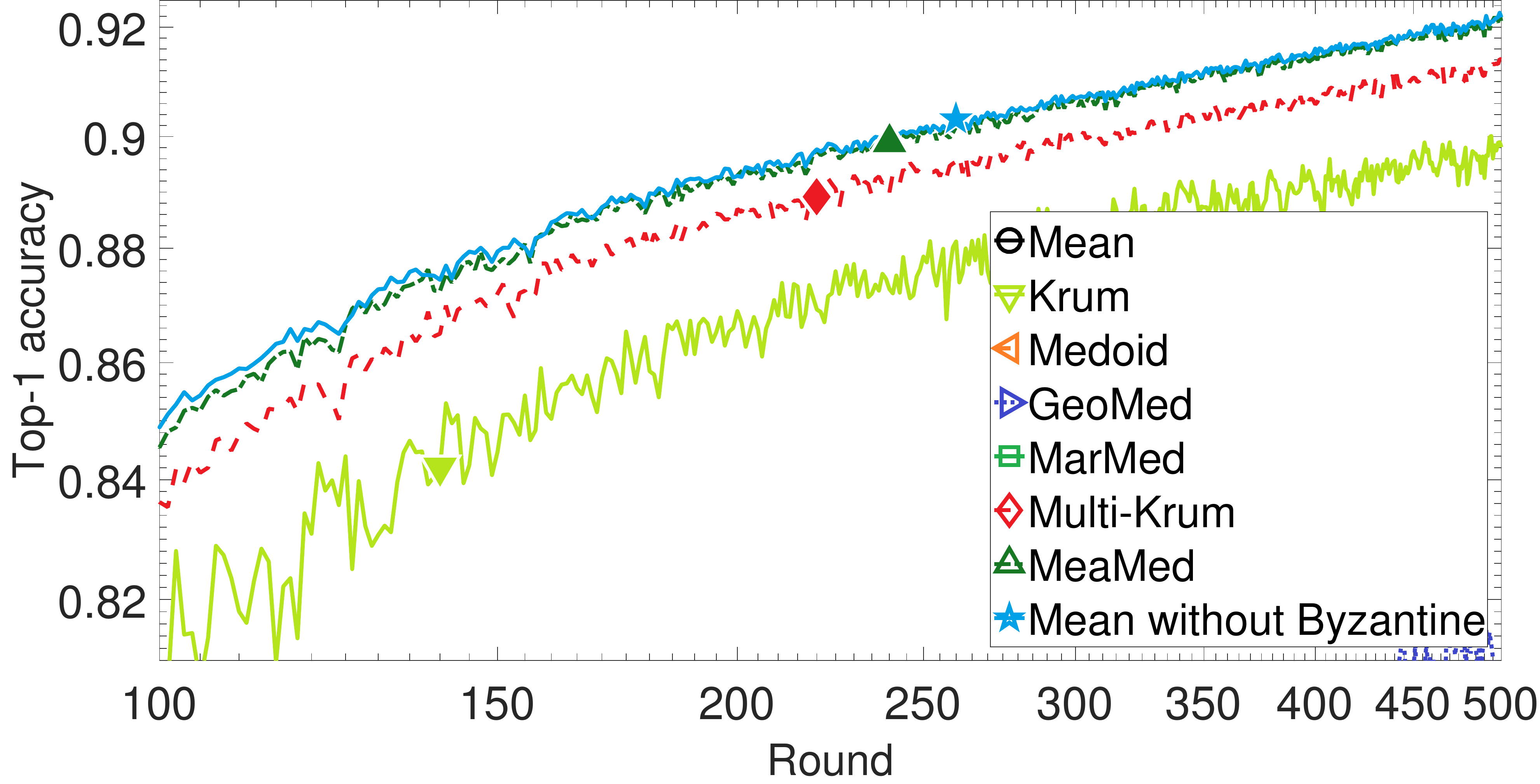}}
\caption{Top-1 accuracy of MLP on MNIST with Omniscient Attack. 6 out of 20 gradient vectors are replaced by the negative sum of all the correct gradients, scaled by a large constant~(1e20 in the experiments). }
\label{fig:mnist_omniscient_appendix}
\end{figure*}
\begin{figure*}[htb!]
\centering
\subfigure[MLP on MNIST with bit-flip]{\includegraphics[width=0.90\textwidth]{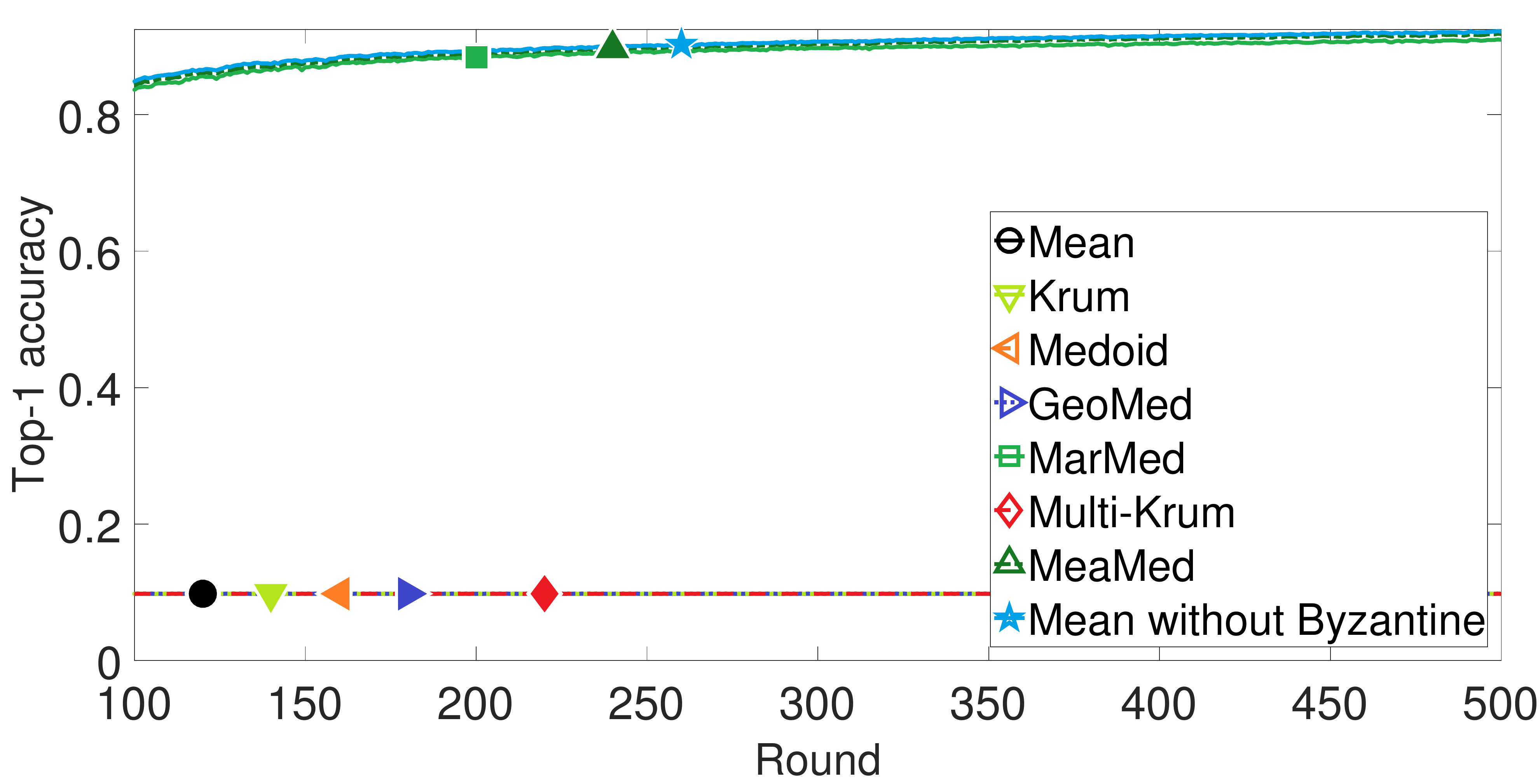}}
\subfigure[MLP on MNIST with bit-flip~(zoomed)]{\includegraphics[width=0.90\textwidth]{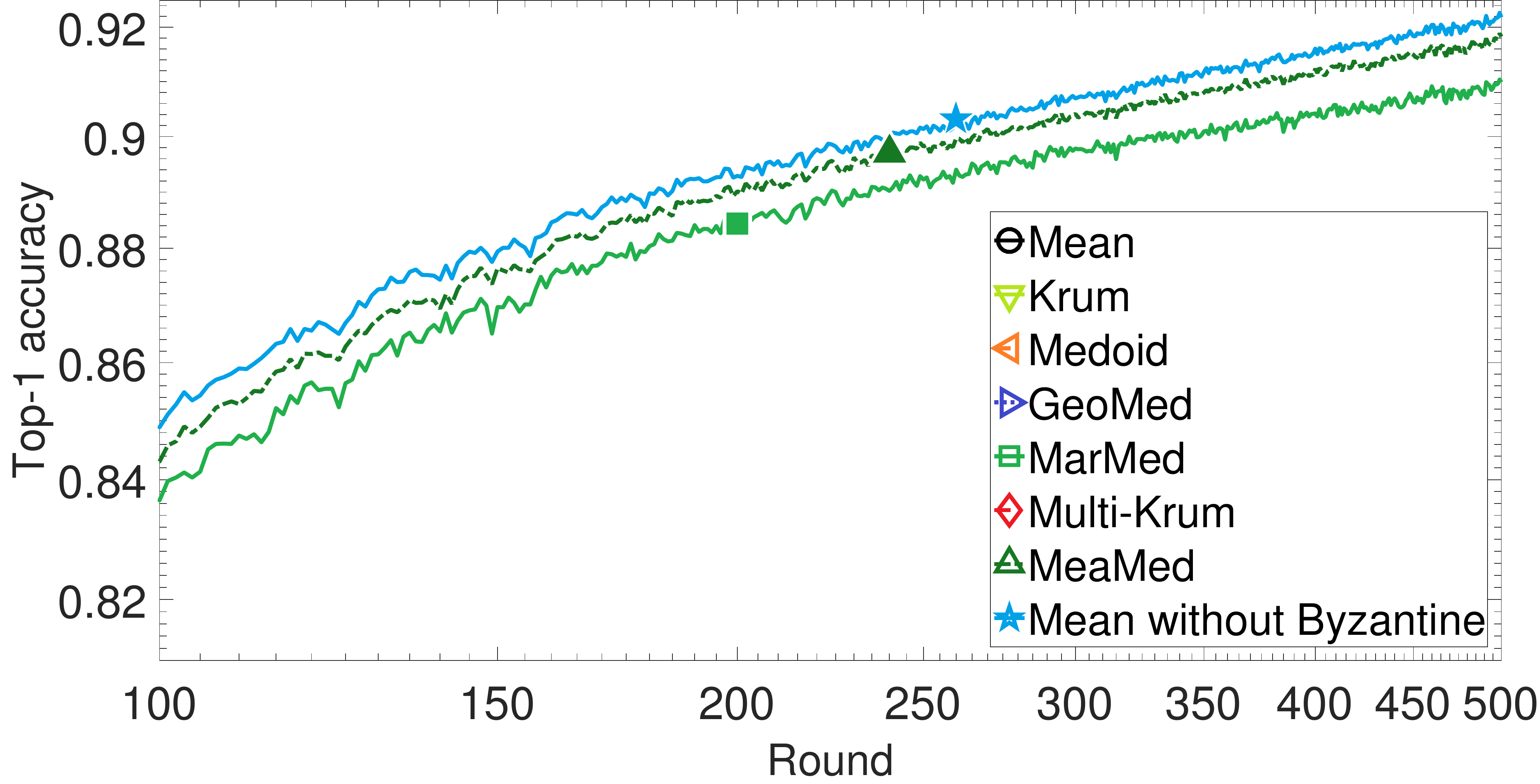}}
\caption{Top-1 accuracy of MLP on MNIST with Bit-flip Attack. For the first 1000 dimensions, 1 of the 20 floating numbers is manipulated by flipping the 22th, 30th, 31th and 32th bits.}
\label{fig:mnist_bitflip_appendix}
\end{figure*}
\begin{figure*}[htb!]
\centering
\subfigure[MLP on MNIST with gambler]{\includegraphics[width=0.90\textwidth]{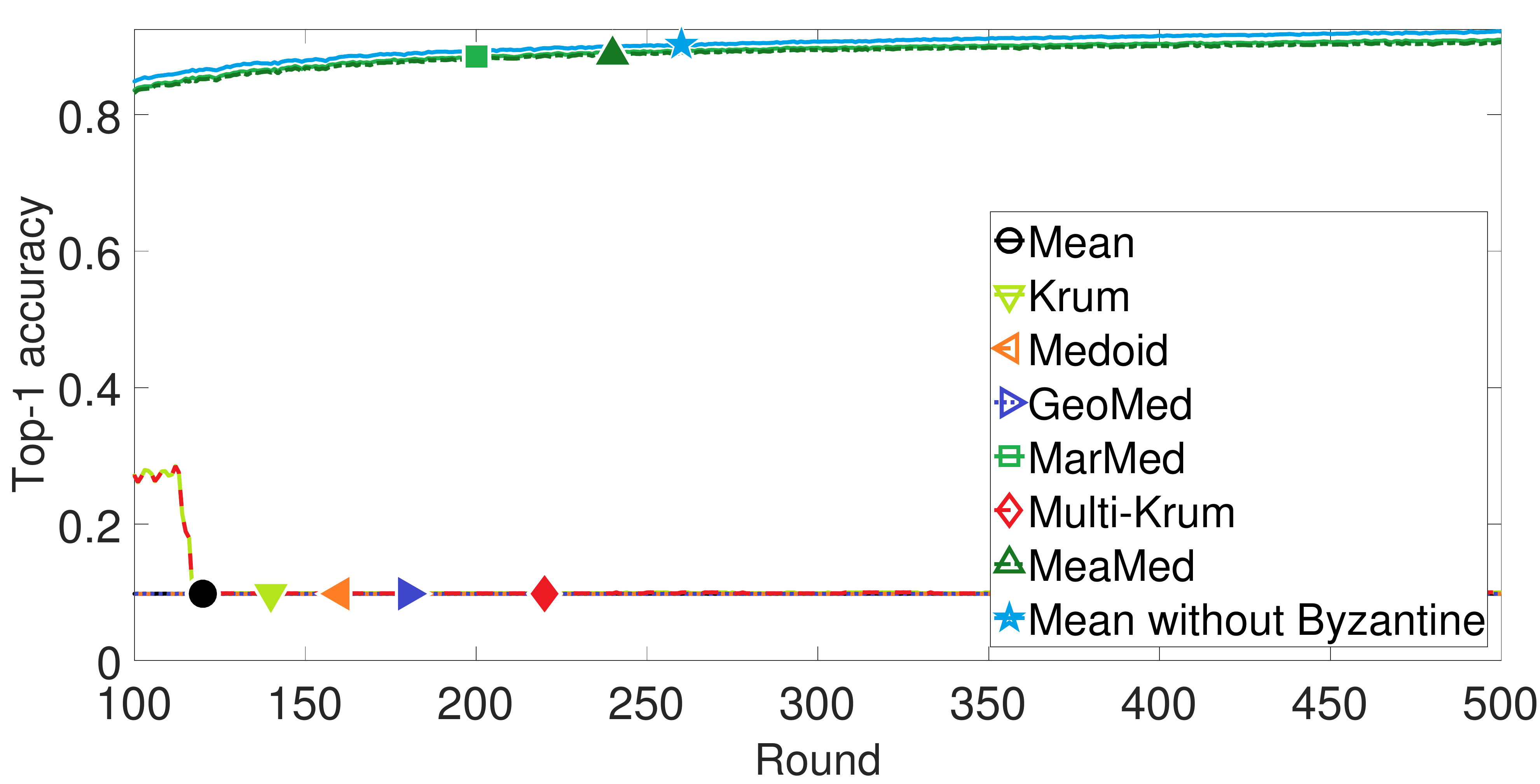}}
\subfigure[MLP on MNIST with gambler~(zoomed)]{\includegraphics[width=0.90\textwidth]{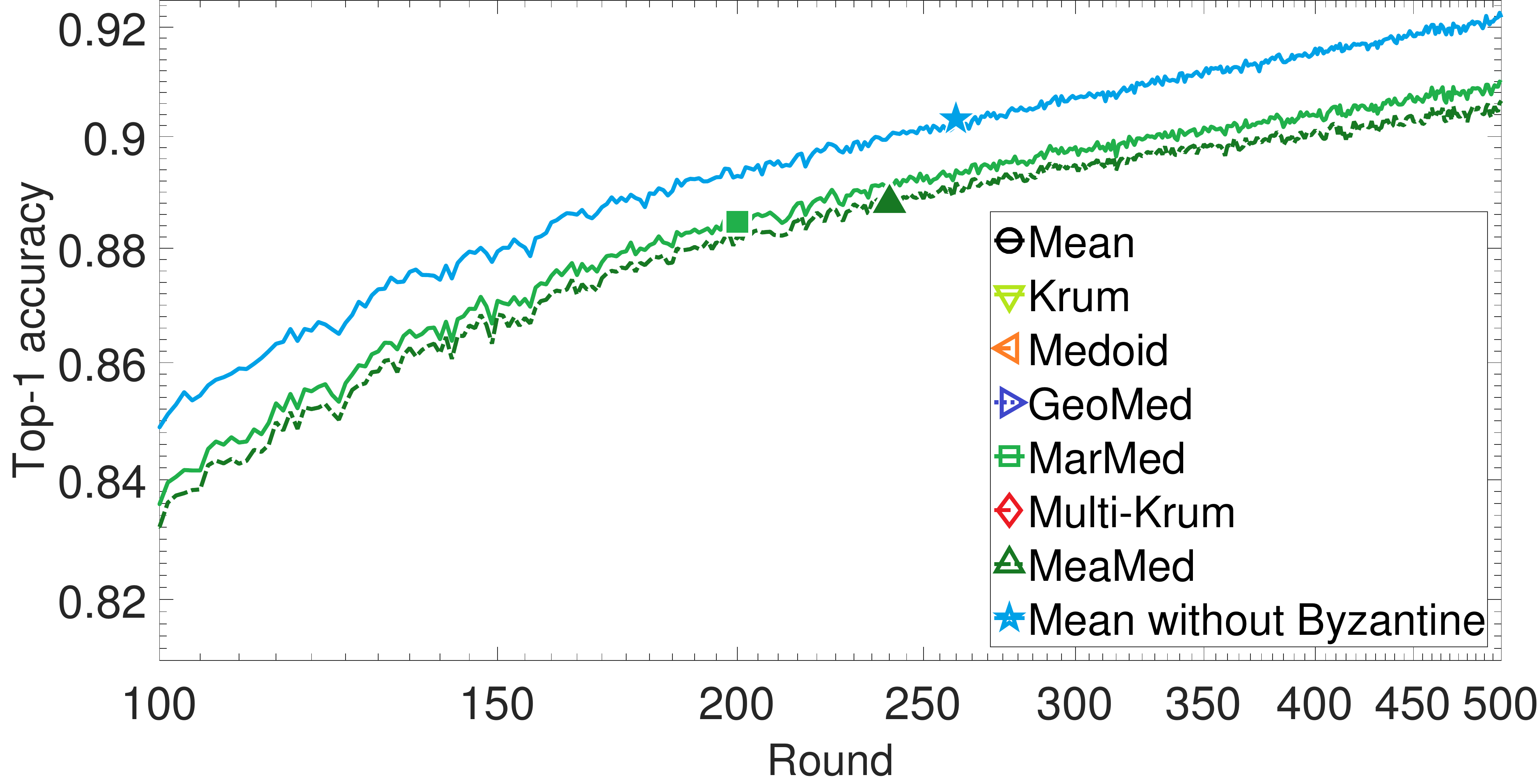}}
\caption{Top-1 accuracy of MLP on MNIST with gambler attack. The parameters are evenly assigned to 20 servers. For one single server, any received value is multiplied by $-1e20$ with probability 0.05\%.}
\label{fig:mnist_multiserver_appendix}
\end{figure*}

\begin{figure*}[htb]
\centering
\includegraphics[width=0.90\textwidth]{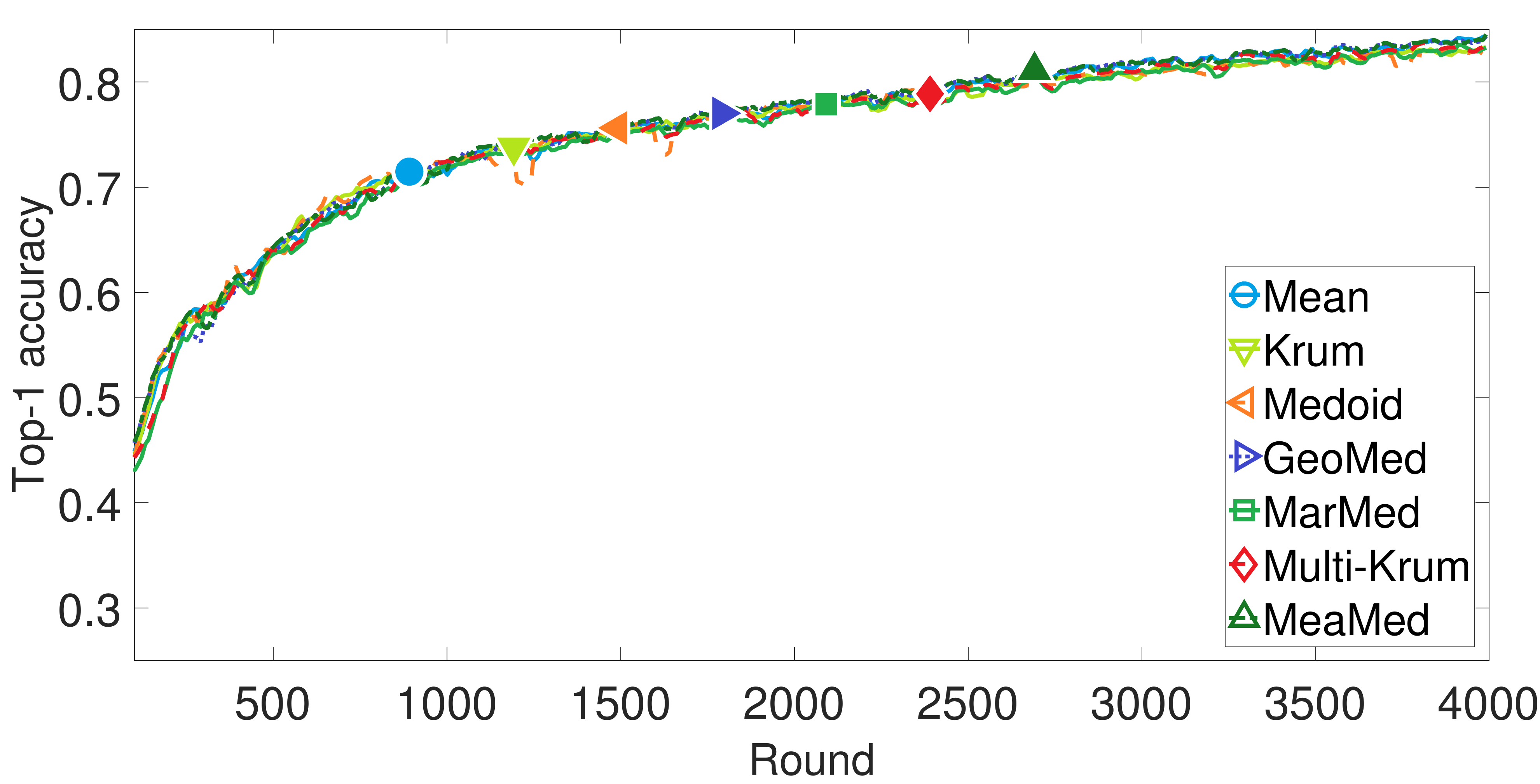}
\caption{Top-3 Accuracy of CNN VS. \# rounds evaluated on CIFAR10 without Byzantine failures}
\label{fig:cifar10_nobyz_appendix}
\end{figure*}
\begin{figure*}[htb]
\centering
\includegraphics[width=0.90\textwidth]{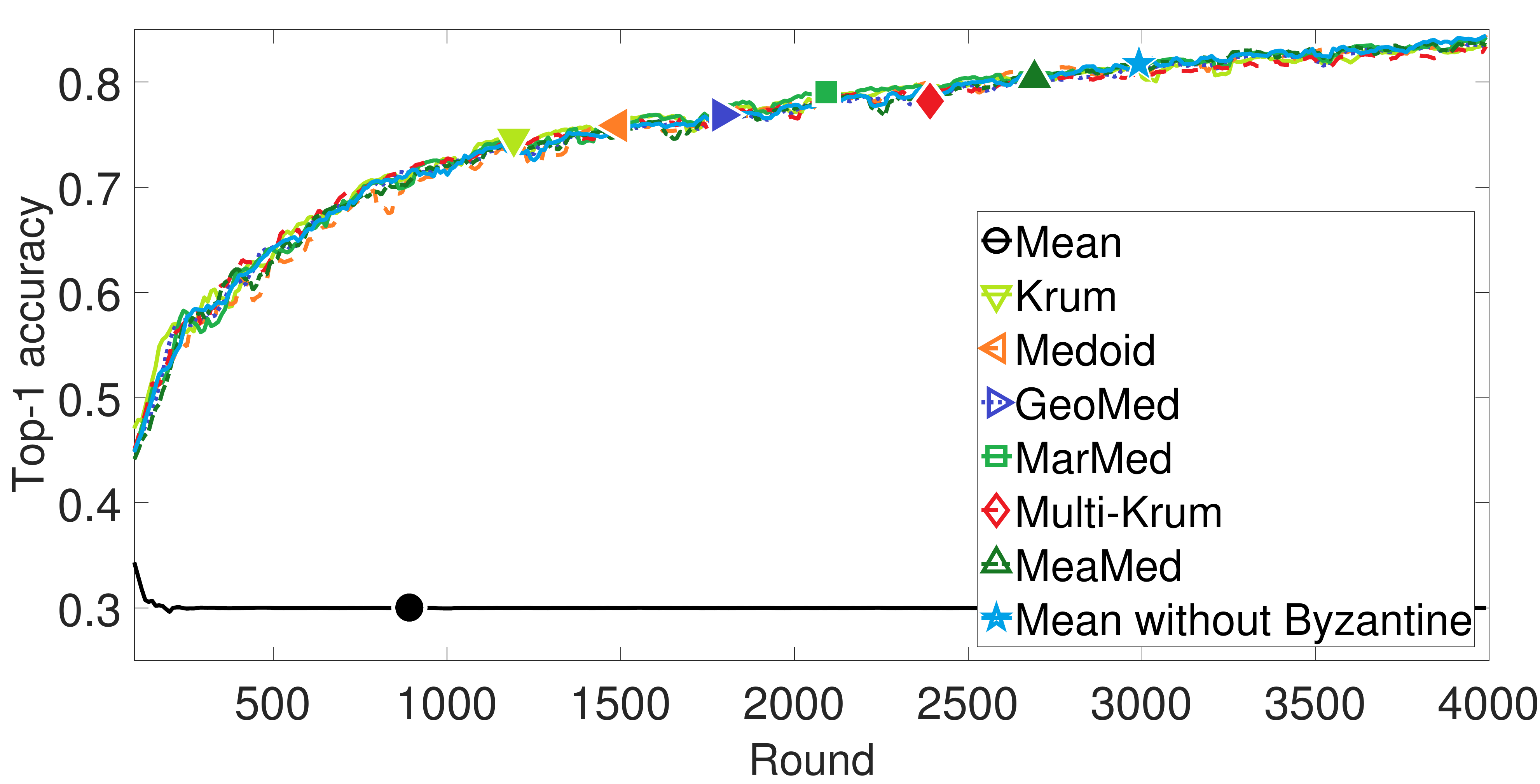}
\caption{Top-3 Accuracy of CNN VS. \# rounds evaluated on CIFAR10 with Gaussian Attack}
\label{fig:cifar10_gaussian_appendix}
\end{figure*}
\begin{figure*}[htb]
\centering
\includegraphics[width=0.90\textwidth]{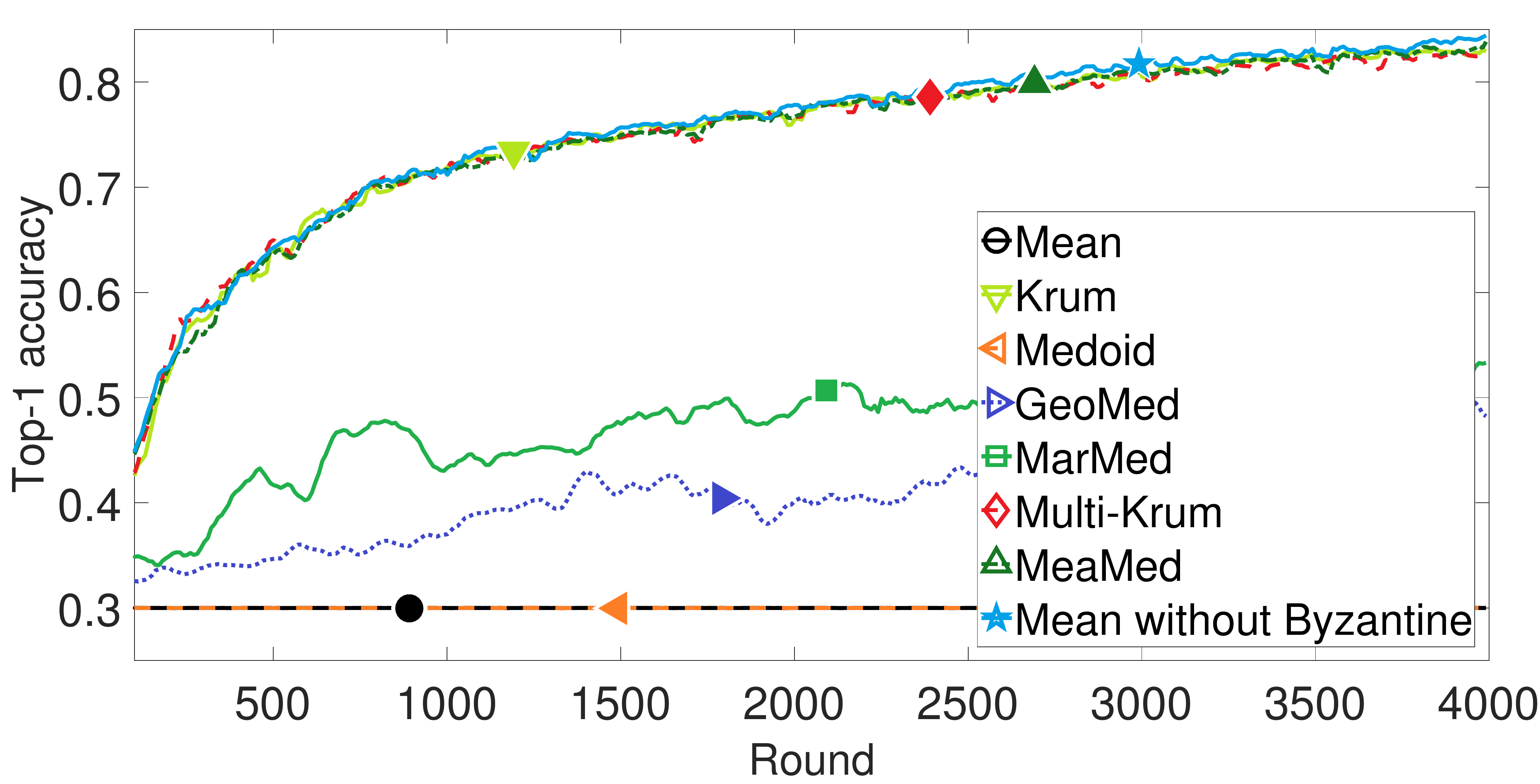}
\caption{Top-3 Accuracy of CNN VS. \# rounds evaluated on CIFAR10 with Omniscient Attack}
\label{fig:cifar10_omniscient_appendix}
\end{figure*}
\begin{figure*}[htb]
\centering
\includegraphics[width=0.90\textwidth]{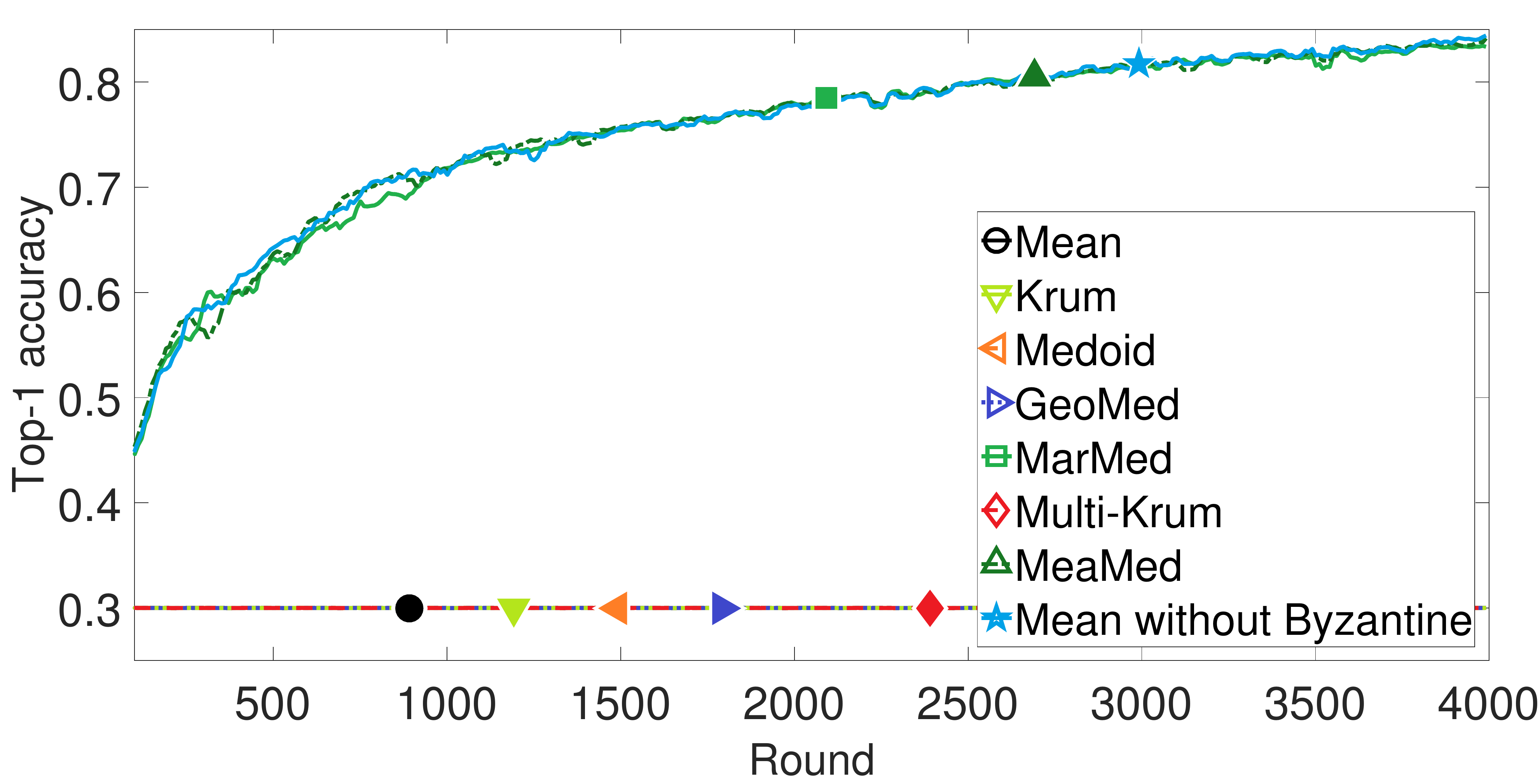}
\caption{Top-3 Accuracy of CNN VS. \# rounds evaluated on CIFAR10 with Bit-flip Attack}
\label{fig:cifar10_bitflip_appendix}
\end{figure*}
\begin{figure*}[htb]
\centering
\includegraphics[width=0.90\textwidth]{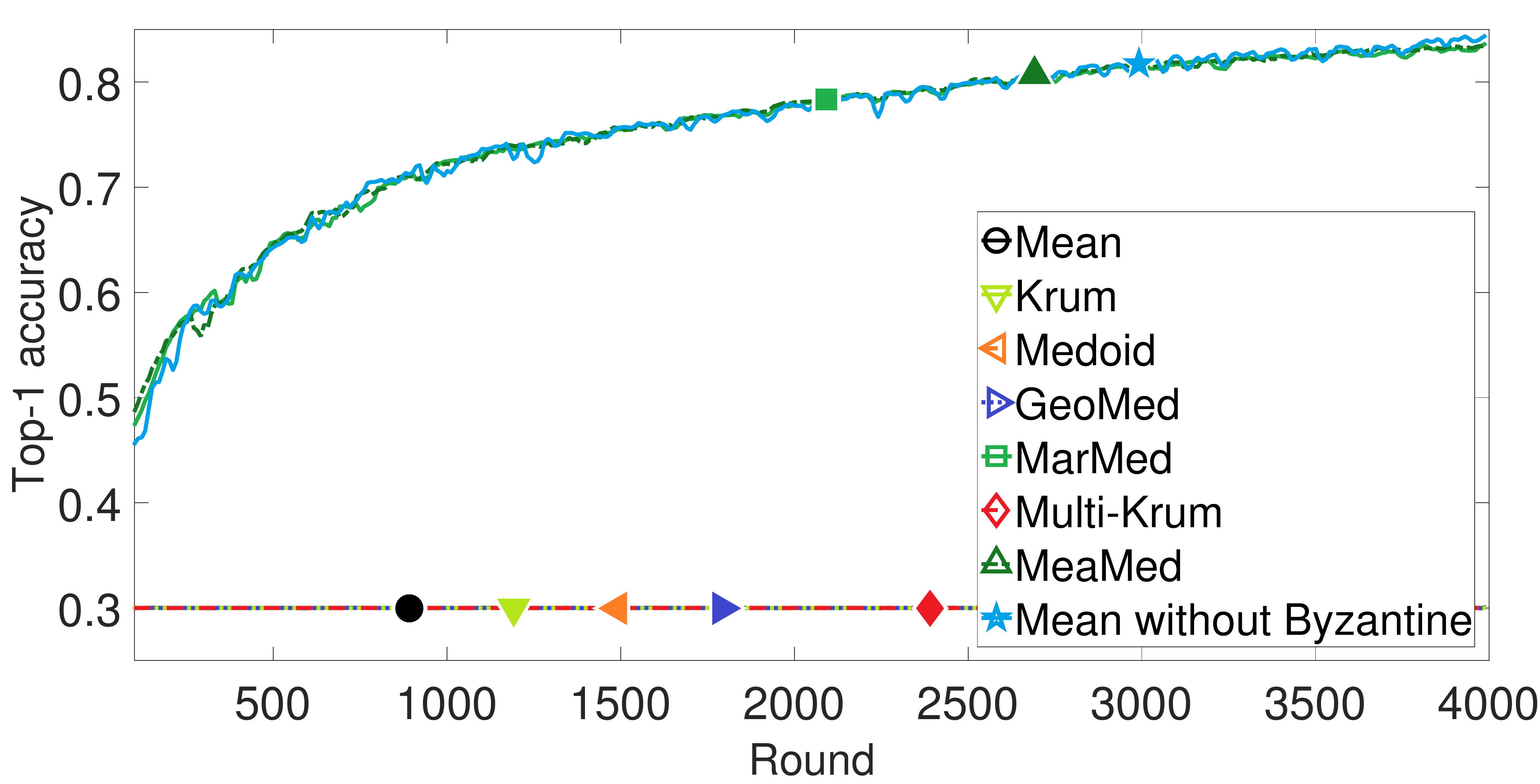}
\caption{Top-3 Accuracy of CNN VS. \# rounds evaluated on CIFAR10 with Gambler attack.}
\label{fig:cifar10_multiserver_appendix}
\vspace{-0.5cm}
\end{figure*}

%\appendix

\end{document}